\documentclass[prx,twocolumn,showpacs,amsmath,amssymb,superscriptaddress,floatfix,longbibliography,nofootinbib]{revtex4-2}
\usepackage{times}
\usepackage{helvet}
\usepackage{courier}
\usepackage{hyperref}
\hypersetup{
  hypertexnames=false, % use unique link anchors for deferred REVTeX floats
  colorlinks=true,    % false: boxed links; true: colored links
  linkcolor=cyan,     % color of internal links
  citecolor=magenta,    % color of links to bibliography
  filecolor=magenta,   % color of file links
  urlcolor=cyan,      % color of external links
  runcolor=cyan
}
\usepackage{graphicx}
\usepackage{orcidlink}
\usepackage{algorithm}
\usepackage{algorithmic}
\usepackage{newfloat}
\usepackage{listings}
\usepackage{amsmath}
\usepackage{amsfonts}
\usepackage{amssymb}
\usepackage{amsthm}
\usepackage{braket}
\usepackage[capitalize]{cleveref}
\crefname{mytheorem}{Theorem}{Theorems}
\Crefname{mytheorem}{Theorem}{Theorems}
\newcommand{\PRR}{\mathrm{PRR}}
\usepackage{xcolor}
\usepackage{booktabs}
\usepackage{tikz}
\usepackage{tikz-cd}
\usepackage{mathtools}
\usepackage{subfigure}
\usetikzlibrary{arrows.meta, positioning}

\newtheorem{mytheorem}{Theorem}

\newcommand{\bes} {\begin{subequations}}
\newcommand{\ees} {\end{subequations}}

\def\>{\rangle}
\def\<{\langle}

\newcommand{\ketb}[2]{|{#1}\>\!\<#2|}

\begin{document}

\title{Trading Circuit Depth for Pulse Sparsity in Chromatic Dynamical Decoupling}

\author{Amy F. Brown\,\orcidlink{0000-0001-6664-6494}}
\affiliation{Department of Physics \& Astronomy, University of Southern California,
Los Angeles, California 90089, USA}
\affiliation{Center for Quantum Information Science \& Technology, University of
Southern California, Los Angeles, California 90089, USA}
\author{Daniel A. Lidar\,\orcidlink{0000-0002-1671-1515}}
\affiliation{Department of Physics \& Astronomy, University of Southern California,
Los Angeles, California 90089, USA}
\affiliation{Center for Quantum Information Science \& Technology, University of
Southern California, Los Angeles, California 90089, USA}
\affiliation{Department of Electrical \& Computer Engineering, University of Southern California,
Los Angeles, California 90089, USA}
\affiliation{Department of Chemistry, University of Southern California, Los Angeles,
California 90089, USA}
\affiliation{Quantum Elements, 2829 Townsgate Road, Westlake Village, California 91361, USA}

\begin{abstract}
Suppressing decoherence and crosstalk systematically across large networks of qubits is a pressing concern as qubit counts increase rapidly.
General multi-qubit dynamical decoupling (DD) has historically been based on Hadamard matrices and orthogonal arrays, with instantaneous-pulse sequences whose circuit depth scales linearly with the number of qubits.
Chromatic-Hadamard DD (CHaDD) improves upon these methods by properly coloring the hardware graph and assigning to each color a row of a Hadamard matrix, resulting in a circuit depth that scales linearly with the number of colors $C$, which can be as low as the chromatic number of the graph.
Here, we explore the tradeoff between circuit depth and the pulse repetition rate (PRR), the average fraction of time steps in which a qubit is pulsed, i.e., a measure of pulse density.
We introduce two single-axis sequences based on binary and Gray matrices, Chromatic-Binary DD (CBDD) and Chromatic-Gray DD (CGDD), which, for $C>2$, are special cases of non-minimum-depth CHaDD and achieve $\PRR<2/C$ and $\PRR=1/C$, respectively, compared with a PRR above $1/4$ for CHaDD, at the price of a circuit depth of $2^C$.
In state preservation experiments on a $127$-qubit IBM processor with $C=3$, the less dense sequences outperform CHaDD at short times when the pulses are not robust; this advantage is largely eliminated by robust versions designed to compensate for coherent pulse errors, and among the robust sequences, CGDD uses the fewest pulses over a fixed protection time.
Across non-robust $C=3$ and $C=5$ schedules, equal PRR on a fixed qubit subset is empirically associated with similar state preservation.
Thus, the PRR serves as a practical figure of merit for chromatic DD whenever coherent pulse errors dominate.
\end{abstract}

\maketitle

\section{Introduction}

Quantum computing efforts in research and industry are rapidly scaling up qubit counts while exploring a widening variety of architectures and qubit connectivities.
Error suppression methods must scale along with the hardware, lest they become the bottleneck.
Generic methods that apply across architectures and require only single-qubit control are of particular interest.

Dynamical decoupling (DD)~\cite{Viola:98,Viola:99,Zanardi:1999fk} is an error suppression method with roots in nuclear magnetic resonance~\cite{Hahn:50,CPMG1958,Maudsley:1986ty} that has experienced a resurgence of interest with the advent of programmable, noisy quantum computers~\cite{Preskill2018} (for a review, see Ref.~\cite{Suter:2016aa}). In particular, it has enabled quantum speedup demonstrations~\cite{pokharel2022demonstration} and improved algorithmic performance in a variety of settings~\cite{jurcevicDemonstrationQuantumVolume2021,baumer2023efficient, Baumer2024}.
DD controls the evolution of a system by applying pulses targeting its qubits at specific times. Conjugation by a pulse that anticommutes with a Hamiltonian error term reverses the sign of that term, and properly scheduled sign reversals average the targeted decoherence and crosstalk terms to zero to leading order, provided the pulse durations and intervals are sufficiently short on the timescale over which the bath evolves.
DD sequences are commonly inserted during idle periods when no other operations occur to avoid interference with computation and measurement~\cite{Ng:2011dn}. Used in this manner, DD has also played an increasingly significant supporting role in recent demonstrations of quantum error correction and fault tolerance~\cite{Google-surface-code:22,Postler:2024aa,Acharya:2025aa,vezvaee2025surfacecodescalingheavyhex}.

\begin{figure*}[t]
\centering
(a)
\includegraphics[width=0.98\columnwidth]{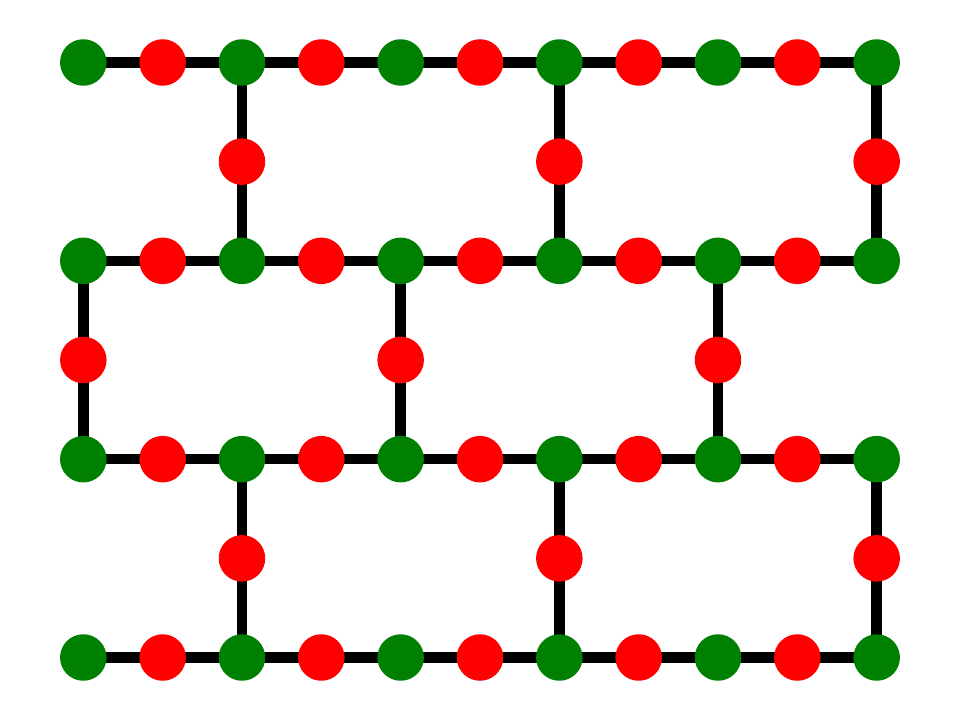}
(b)
\includegraphics[width=0.98\columnwidth]{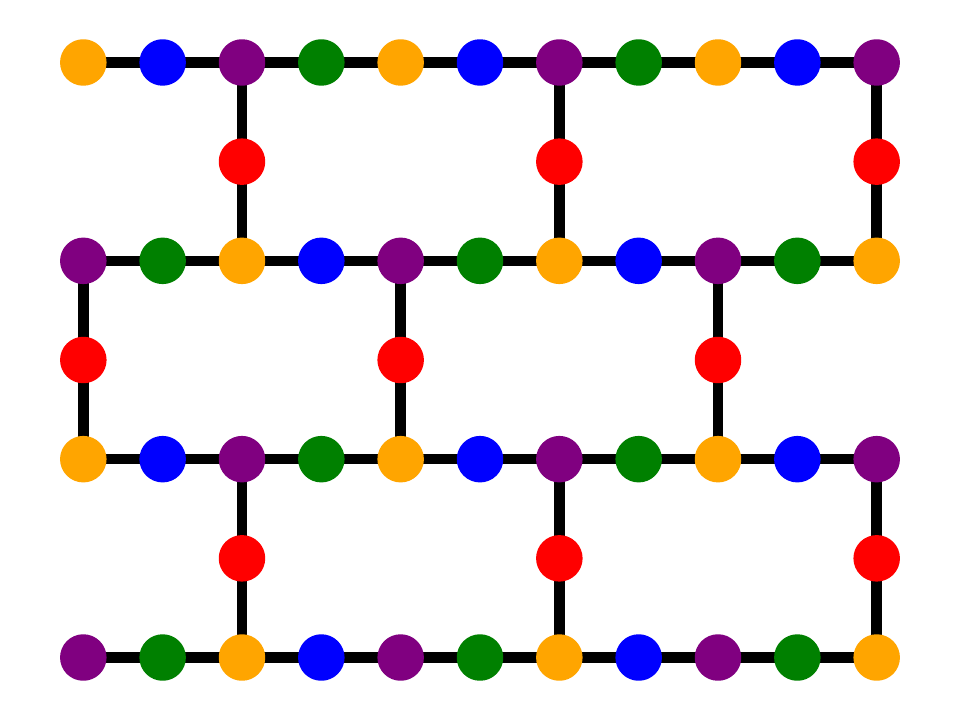}
\\
(c)
\includegraphics[width=0.63\columnwidth]{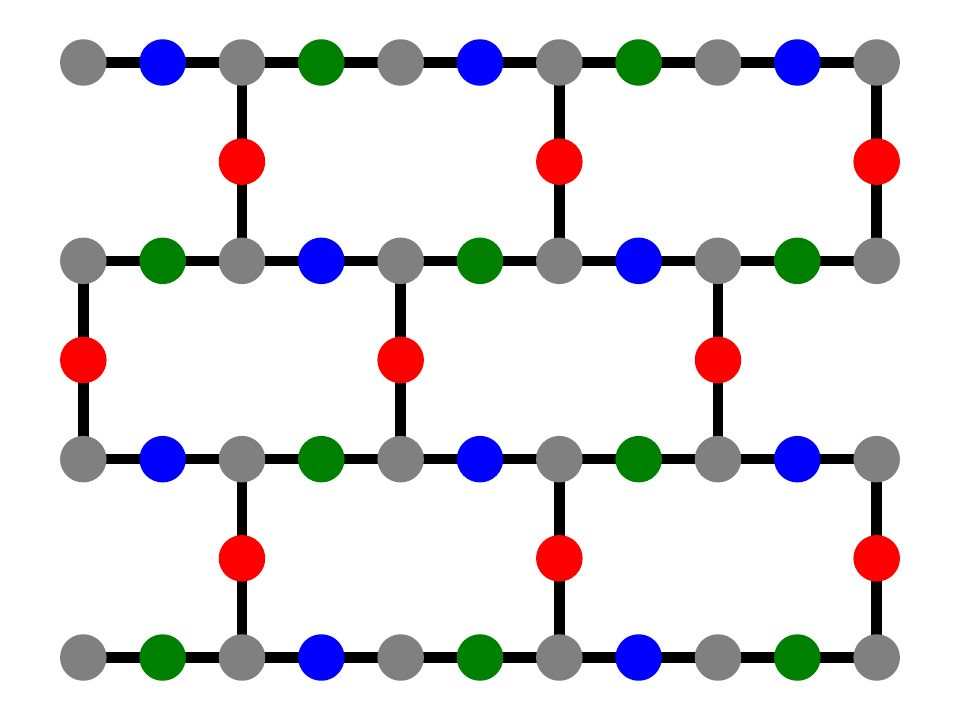}
(d)
\includegraphics[width=0.63\columnwidth]{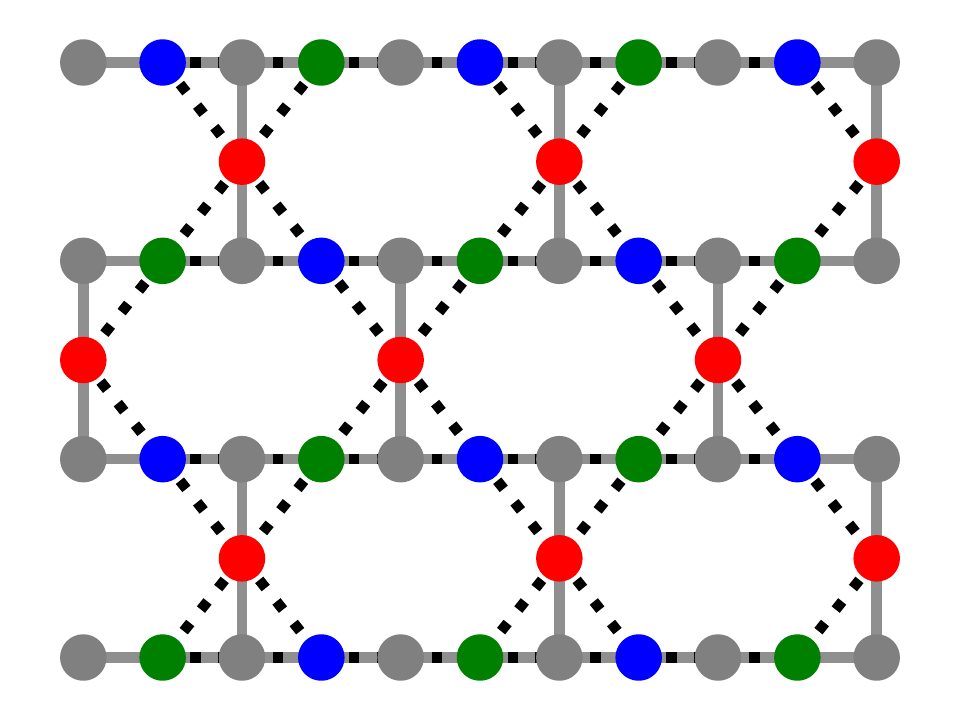}
(e)
\includegraphics[width=0.63\columnwidth]{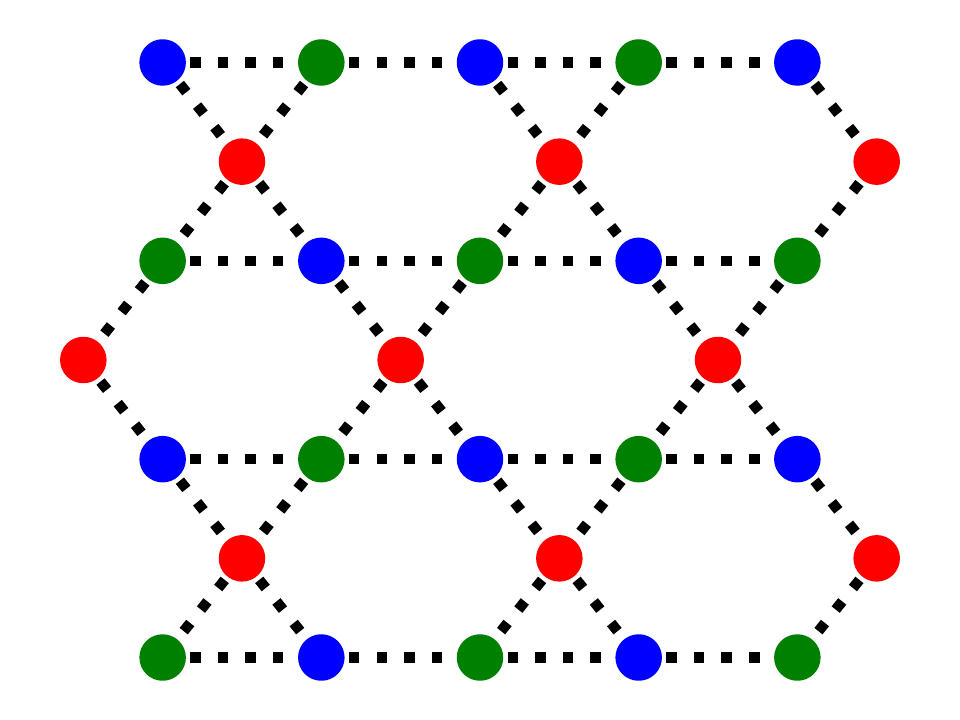}
\caption{
(a)
$C=2$: Minimum proper (distance-1 2-)coloring of a representative section of the IBM heavy-hex lattice, with the edge qubits of the underlying hexagonal lattice (bivalent in the bulk) colored red and the vertex qubits (trivalent in the bulk) colored green.
(b)
$C=5$: Distance-3 5-coloring of the same section.
(c)
$C=3$: Proper 3-coloring of the active red, green, and blue (RGB) interaction subgraph embedded in the IBM heavy-hex lattice; gray qubits are spectators and are not counted among the $C$ colors.
Note that the RGB qubits are the same in the $C=5$ and $C=3$ colorings.
(d)
Solid gray edges [black in (a-c)] represent intentional nearest-neighbor couplings corresponding to (distance-1) edges in the hardware graph;
dotted black (distance-2) edges represent all next-nearest-neighbor couplings among the RGB qubits, which are not intentionally engineered but are often measurable and can be suppressed.
(e)
Removing the gray qubits and any edges incident on them from the graph results in an embedded, properly 3-colored kagome lattice.
}
\label{fig:heavy-hex-embedded-3coloring}
\end{figure*}

The majority of DD research has focused on the single-qubit setting, where the theory is mature and includes high-order sequences~\cite{Uhrig:2007qf} as well as sequences that are robust to pulse imperfections~\cite{Souza:2012aa,Genov2017}.
Much less attention has been devoted to the multi-qubit setting, where an important role is played by residual two-body couplings resulting in crosstalk~\cite{Sarovar2020detectingcrosstalk}, such as the always-on $ZZ$ interaction between neighboring superconducting qubits, which acts even while the qubits idle~\cite{Ku:2020aa,Tripathi:2022aa}.
Identical sequences applied to every qubit cancel single-qubit decoherence but commute with these couplings, leaving them untouched. To resolve this, the pulses applied to different qubits must be scheduled relative to one another.
Most multi-qubit results have been based on certain sign matrices, specifically, Hadamard matrices and orthogonal arrays~\cite{Leung:01,Stollsteimer:01,Rotteler:2006aa,Wocjan:2006aa,Bookatz:2016aa}, with instantaneous-pulse constructions whose circuit depth scales linearly with the number of qubits.
The syncopated DD of Ref.~\cite{evert2024syncopated} instead schedules mutually syncopating periodic sequences (related by time shifts, frequency multiplication, or operator alternation) on the color classes of the crosstalk graph, one sequence per color; the minimal length of syncopated sequences needed to cancel all $ZZ$-type crosstalk was found to be $2^{\chi}$, where $\chi$ is the chromatic number of the graph.

By combining graph coloring with Hadamard matrices, we showed in Ref.~\cite{brown2024efficient} that Chromatic-Hadamard DD (CHaDD) decouples an entire hardware graph in a number of time steps that scales linearly with the chromatic number rather than with the number of qubits, for both single-axis and multi-axis decoupling (with a smaller prefactor in the single-axis case); for hardware graphs with a constant chromatic number, such as the heavy-hex lattice ($\chi=2$), the depth is independent of the number of qubits.
More recently, graph coloring was combined with classical error-detecting codes to construct sequences that cancel selected $k$-local interactions~\cite{Nguyen2026color}.
CHaDD has also been used in logical state preservation experiments to mitigate dephasing and crosstalk~\cite{Joshi2026demonstratingnoiseadaptedquantumerror}.
Ref.~\cite{Kim2026highorder} has extended CHaDD to higher order in the weak-coupling regime: on a graph of chromatic number $\chi$, sequences of $O(\chi p)$ multi-qubit pulse events can cancel every error term linear in the system-bath coupling strength through order $p$ in the total evolution time, leaving only errors of order $p+1$ in time or of second order in the coupling strength.
Several related approaches schedule or compile crosstalk-suppressing DD at the circuit level~\cite{Mundada:2023aa,Zhou:2023aa,Shirizly:2024aa,CarreraVazquez:2024aa,Niu:2024aa,seif2024suppressing,Coote:2025aa,Hickman:2025aa}.

A chromatic DD sequence is characterized by four figures of merit: (i) the set of Hamiltonian terms it cancels, to a given order in the pulse interval; (ii) its circuit depth $N$, i.e., the number of time steps per sequence repetition, which must fit within the available idle window; (iii) its pulse repetition rate (PRR), i.e., the average fraction of time steps in which a qubit is pulsed; and (iv) its robustness to pulse imperfections.
The first is common to every chromatic sequence we consider here: all are single-axis, i.e., they involve only $X$ pulses, and cancel to first order every one- and two-body term except those of $X$ and $XX$ type (\cref{thm:chadd,thm:cbdd,thm:cgdd}). This is a relevant setting for fixed-coupler superconducting quantum processing units (QPUs), where dephasing and always-on $ZZ$ crosstalk are important idle errors~\cite{Tripathi:2022aa}; we restrict attention throughout to this first-order setting, and leave multi-axis and higher-order chromatic sequences outside the scope of this work.
The second was the focus of Ref.~\cite{brown2024efficient}.
The third and fourth are linked: the number of pulses a qubit accumulates over a fixed protection time is proportional to the PRR, and systematic pulse imperfections contribute coherent errors that accumulate unless the sequence is designed to compensate for them. These are the aspects we focus on in this work.

CHaDD attains a circuit depth of at most $2C$ for any properly $C$-colored graph, but at the price of a high pulse density: even for the pulse-count-minimizing assignment of colors to Hadamard rows, the PRR exceeds $1/4$ (\cref{thm:chadd}), and for $C=3$ it equals $2/3$, i.e., each qubit is pulsed in two thirds of the time steps on average.
Whether the depth can be traded for a lower pulse density within the same sign-matrix framework, and whether such a trade improves state preservation on hardware, was left open.

Here, we address this question and introduce Chromatic-Binary DD (CBDD) and Chromatic-Gray DD (CGDD), two sparser schedules (except at $C=1,2,4$) whose nonconstant sign rows are rows of the $2^C$-column Hadamard matrix, so that the CHaDD averaging argument applies as is.
The main results of this work are as follows (\cref{thm:cbdd,thm:cgdd}): CBDD and CGDD cancel the same terms as CHaDD to first order, with $\PRR<2/C$ and $\PRR=1/C$, respectively; the price is a circuit depth of $2^C$, i.e., exponential rather than linear in the number of colors.
More generally, only two properties of the sign matrix are essential: each color-assigned row must add up to zero, and every pair of assigned rows must be orthogonal.
We use this observation to characterize the control groups of all three sequences, to identify the row assignment that minimizes the CHaDD pulse count, which we call chromatic-Walsh DD (CWDD) and whose PRR bounds enter \cref{thm:chadd}, and to show that CBDD and CGDD additionally cancel a large class of $k$-body terms that CHaDD row assignments need not.

We then compare the sequences on the 127-qubit \texttt{ibm\_strasbourg} QPU, using the embedded $C=3$ coloring of the heavy-hex lattice [\cref{fig:heavy-hex-embedded-3coloring}(c)] and state preservation experiments over the six Pauli eigenstates.
At short times, the sparser non-robust sequences perform better, in the reverse order of their PRRs; introducing robust variants, built from the universally robust UR4 pulse block~\cite{Genov2017}, largely erases this advantage.
Both observations support treating the PRR as a proxy for the coherent error accumulated from imperfect pulses.
Grouping non-robust sequences by their PRR on a fixed qubit subset, we further find that equal PRR is empirically associated with similar state preservation, even across different values of $C$. This situates the PRR metric as the key performance predictor in these comparisons.

The structure of this paper is as follows. In \cref{sec:background} we review graph coloring, the noise model, dynamical decoupling, the PRR, and CHaDD.
In \cref{sec:results} we introduce CBDD and CGDD, prove their decoupling properties, and compare their circuit depth and PRR with those of CHaDD and CWDD.
In \cref{sec:experiments} we compare the empirical performance of these sequences in preserving an initial state on an IBM QPU, including a comparison across $C=3$ and $C=5$ sequences of equal RGB PRR. We conclude in \cref{sec:conc}.
Appendix~\ref{app:CHaDD-thm-proof} contains the proof of the CHaDD theorem.

\section{Background}
\label{sec:background}

\subsection{Graph Coloring}
\label{sec:graph-coloring}

A (vertex) $C$-coloring of a graph $G=(V,E)$ is a map $f : V \rightarrow \{ 1, 2, \dots , C\}$
labeling each vertex of the graph with one of $C$ colors.
Throughout, $C$ denotes the number of nonempty color classes, so $f$ may be taken to be surjective.
A coloring $f$ is \textit{proper} if no edge connects two vertices of the same color, i.e., $uv \in E \implies f(u) \neq f(v)$.
The minimum number of colors that can be used to properly color a graph is called the \textit{chromatic number} and denoted $\chi$.
A vertex coloring is \textit{minimum} if it uses $\chi$ colors, i.e., $C=\chi$.
Note that any proper $C$-coloring must satisfy $C \geq \chi$.

The distance $d(u,v)$ between two vertices $u,v \in V$ is the length of the shortest path between them on the graph $G$.
There is a higher-distance analog of proper graph coloring called distance-$d$ coloring, which, for distinct vertices $u,v$, requires that $d(u,v)\leq d \implies f(u) \neq f(v)$, so that two vertices of the same color are at distance at least $d+1$, i.e., $f(u)=f(v) \implies d(u,v) > d$.
Proper coloring is the distance-1 case, since $uv \in E \iff d(u,v) = 1$.

In this work, we focus on the IBM heavy-hex lattice~\cite{Chamberland:2020aa}, which has chromatic number $\chi=2$; its minimum proper $\chi$-coloring is depicted in \cref{fig:heavy-hex-embedded-3coloring}(a).
We also consider colorings with more colors than the lattice strictly requires, i.e., $C>\chi$.
The first is a distance-3 5-coloring [\cref{fig:heavy-hex-embedded-3coloring}(b)], chosen so that no two qubits of the same color are within three edges of one another; this more than suffices to cover the next-nearest-neighbor couplings of \cref{fig:heavy-hex-embedded-3coloring}(d), which connect qubits at distance 2.
The second is a proper 3-coloring of the active red, green, and blue (RGB) interaction subgraph depicted in \cref{fig:heavy-hex-embedded-3coloring}(c), which is minimum for that subgraph and results in an effective kagome lattice as shown in \cref{fig:heavy-hex-embedded-3coloring}(e).
In this embedded construction, the graph $G$ used in the decoupling analysis has the RGB qubits as its vertex set; the gray qubits are spectators outside the $C=3$ color count, and their couplings to the RGB system can be incorporated into the bath operators of the noise model we describe next.

\subsection{Noise Model}
\label{sec:noise-model}

The $K$-local Hamiltonian $H$ of a quantum device is a sum of $k$-body Hamiltonians $H_k$,
\begin{align}
\label{eq:k-local-hamiltonian}
    H = \sum_{k=1}^K H_k,
\end{align}
where the $k$-body Hamiltonian $H_k$ is the sum of all Pauli terms with support on $k$ qubits:
\bes
\begin{align}
\label{eq:k-body-hamiltonian-decomposition}
    H_k &= \sum_{\vec\alpha \in \{x,y,z\}^k} H_k^{\vec\alpha} , \\
    \label{eq:k-body-hamiltonian}
    H_k^{\vec\alpha} &= \sum_{\substack{v_1,v_2,\dots,v_k \in V\\ v_1<v_2<\dots<v_k}} \sigma_{v_1}^{\alpha_1} \otimes \sigma_{v_2}^{\alpha_2} \otimes \dots \otimes \sigma_{v_k}^{\alpha_k} \otimes B_{\vec v_k}^{\vec\alpha} ,
\end{align}
\ees
where $\vec v_k = (v_1,v_2,\dots,v_k)$.
The operators $B_{\vec v_k}^{\vec\alpha}$ represent either a scalar multiple of the identity operator, in the case of unwanted system interactions (such as crosstalk), or bath operators acting purely outside the system; these give rise to closed- and open-system evolution, respectively.
We can restrict the sum in \cref{eq:k-body-hamiltonian} to $\vec v_k \in E_k$, where $E_k$ is a set of canonically ordered $k$-vertex hyperedges, to limit which interactions are included.

For simple graphs, we typically consider only $K=2$-local Hamiltonians, for which the $k=1$- and 2-body Hamiltonians are conveniently represented as sums over the vertices $V$ and edges $E$ of the graph $G=(V,E)$, respectively.
Such a Hamiltonian representing the graph $G$ is denoted $H_G$ and decomposed exactly as in Ref.~\cite{brown2024efficient} as 
\bes
\begin{align}
    \label{eq:graph-hamiltonian}
    H_G &= H_1 + H_2
    , \\
    \label{eq:1-body-hamiltonian}
    H_1 &= \sum_{\alpha \in \{x,y,z\}} H_1^\alpha
    , \hspace{5mm}
    H_1^\alpha = \sum_{v \in V} \sigma_v^\alpha \otimes B_v^\alpha
    , \\
    \label{eq:2-body-hamiltonian}
    H_2 &= \sum_{\alpha,\beta \in \{x,y,z\}} H_2^{\alpha\beta}
    , \hspace{5mm}
    H_2^{\alpha\beta} = \sum_{\substack{uv \in E\\ u<v}} \sigma_u^\alpha \sigma_v^\beta \otimes B_{uv}^{\alpha\beta}
    .
\end{align}
\ees
The edge set $E$ in \cref{eq:2-body-hamiltonian} includes both intentionally engineered couplings and unintentional couplings that are nevertheless measurable. A pure-bath term $I\otimes H_B$ may be added to $H_G$ without affecting any of the averaging arguments below, since it commutes with every control unitary; we omit it for brevity.
For later reference, let $f_\tau=e^{-i\tau H_G}$ denote the free (uncontrolled) evolution for a time duration $\tau$.

Every graph coloring partitions the vertex set by color, and every proper coloring partitions the edge set into ordered pairs of distinct colors:
\bes
\begin{align}
    \label{eq:graph-coloring-vertex-partition}
    V_c &\equiv \{ v \in V ~|~ f(v)=c \}
    , \\
    \label{eq:graph-coloring-edge-partition}
    E_{c_1,c_2} &\equiv \{ uv \in E ~|~ u<v, f(u)=c_1, f(v)=c_2 \}
    .
\end{align}
\ees
We can regroup the Hamiltonian terms in \cref{eq:1-body-hamiltonian,eq:2-body-hamiltonian} according to the above partition, with bath operators elided for brevity, as
\begin{align}
    H_1^\alpha = \sum_{c=1}^C \sum_{v \in V_c} \sigma_v^\alpha
    , \hspace{5mm}
    H_2^{\alpha\beta} = \sum_{c_1 \neq c_2} \sum_{uv \in E_{c_1,c_2}} \sigma_u^\alpha \sigma_v^\beta
    .
\end{align}

\subsection{Dynamical Decoupling}
\label{sec:dd}

DD is an open-loop error suppression technique in which pulses are inserted during idle periods and timed so that certain noise terms are canceled up to a prescribed order in time.
Throughout the analysis we adopt the ``bang-bang'' idealization, in which pulses are instantaneous and perfect; finite pulse durations enter in the experiments (\cref{sec:experiments}).
If a system evolves freely for time $\tau$ according to a unitary $f_\tau = e^{-i \tau H}$, conjugating this by a unitary $U_j$ results in the unitary $U_j^\dagger f_\tau U_j$, which is generated by the Hamiltonian $U_j^\dagger H U_j$.
If these conjugated free evolution periods are conjoined for $j=0,1,\dots,N-1$, the Baker-Campbell-Hausdorff formula shows that the total evolution is generated, to first order in time, by the average Hamiltonian~\cite{Waugh:68}
\begin{align}
    \label{eq:decoupling-averaged-hamiltonian}
    \overline{H} = \frac1N \sum_{j=0}^{N-1} U_j^\dagger H U_j .
\end{align}
When the list of conjugating unitaries contains every element of a finite group $\mathcal{G}$ with equal multiplicity, the original Hamiltonian $H$ is projected onto the commutant of $\mathcal{G}$, i.e., $\left[g,\overline{H}\right]=0~\forall~g\in\mathcal{G}$.

\subsection{Pulse Repetition Rate (PRR)}
\label{sec:prr}

Circuit depth $N$ and pulse repetition rate (PRR) were the primary measures of DD sequence efficiency in Ref.~\cite{brown2024efficient}.
For a DD sequence of circuit depth $N$ scheduled on $C$ colors, let $P_c$ denote the number of pulses each qubit of color $c$ undergoes in a single sequence repetition.
The PRR is then
\begin{align}
\label{eq:PRR}
    \PRR = \frac{\sum\limits_{c=1}^C |V_c| P_c}{N \sum\limits_{c=1}^C |V_c|} \approx \frac1{NC}\sum_{c=1}^C P_c  = \frac{P}{NC} ,
\end{align}
where $P=\sum_c P_c$ is the number of pulses across all colors.
In the second expression, we have made the approximation that the color classes are equally populated, which renders the PRR a property of the sequence alone; for simplicity, from here on we adopt this definition, $\PRR=P/(NC)$.
We may also evaluate the PRR on a subset of colors, i.e., on $\mathcal{C}' \subseteq \{1,2,\dots,C\}$, as follows:
\begin{align}
\label{eq:subset-prr}
    \PRR_{\mathcal{C}'} = \frac{\sum\limits_{c\in\mathcal{C}'}|V_c|P_c}{N\sum\limits_{c\in\mathcal{C}'}|V_c|} \approx \frac1{N|\mathcal{C}'|} \sum_{c\in \mathcal{C}'} P_c ,
\end{align}
with the same equal-population approximation.

The PRR is a dimensionless pulse density and is independent of the physical interval $\tau$.
If one pulse is applied to every qubit of every color in every one of the $N$ time steps, then $P_c=N$ for all $c$ and $\PRR=1$; otherwise, the PRR is the fraction of occupied color-time slots.
For uniform time steps, the corresponding average pulse rate per color is $\PRR/\tau$.
We show below that the PRR serves as a proxy for accumulated coherent error in sequences that are not robust to pulse imperfections.

\subsection{Chromatic-Hadamard DD (CHaDD)}
\label{sec:chadd}
Single-axis Chromatic-Hadamard DD (CHaDD)~\cite{brown2024efficient} schedules multi-qubit DD sequences according to a proper coloring of the graph $G$ of a graph Hamiltonian $H_G$ [\cref{eq:graph-hamiltonian}], using a Hadamard matrix of suitable size to assign each color to a distinct row $i>0$ of the matrix.
All qubits of a given color follow the same sequence determined by their assigned row, and qubits of different colors follow sequences that suppress not only the decoherence on each qubit but also the crosstalk between them.

Consider the unnormalized Hadamard sign matrix $W_\nu$ of dimension $N=2^\nu$,
\begin{align}
    \label{eq:hadamard-matrix}
    W_\nu \equiv \sum_{i,j=0}^{N-1} (-1)^{i \cdot j} \ketb{i}{j}
    ,
\end{align}
where $i$ and $j$ are identified with $\nu$-bit strings and the dot denotes their binary inner product. Here and below we use $\ket{j}$ ($\bra{j}$) to denote the $j$th standard column (row) basis vector.
Choosing
$\nu=\lfloor\log_2 C\rfloor+1=\lceil\log_2(C+1)\rceil$
ensures that $W_\nu$ has at least $C$ nonzero row indices.
We assign each color $c=1,\dots,C$ to a distinct row $i>0$ via an injective function $g$, i.e., $i=g(c)$.
Let us denote the simultaneous application of $X$ pulses to all qubits of a given color $c$ as
\begin{align}
\label{eq:xtilde}
    \widetilde{X}_c \equiv \bigotimes_{v \in V_c} X_v
    ,
\end{align}
and define the toggling-frame control unitary in the $j$th time step as
\begin{align}
    \label{eq:chadd-decoupling-unitary}
    U_j = \bigotimes_{c=1}^C \widetilde{X}_c^{g(c) \cdot j}
    .
\end{align}
This mirrors the $j$th column of the Hadamard matrix:
for every color $c=1,\dots,C$ assigned to row $i=g(c)$, $U_j$ contains $\widetilde{X}_c$ if the $(i,j)$th entry of $W_\nu$ is $-1$, i.e., if $i \cdot j \equiv 1 \pmod 2$.

Equivalently, 
qubits of color $c$ undergo an $X$ pulse in time step $j$ if the sign of the entries in row $i=g(c)$ changes between columns $j$ and $(j+1)\bmod N$.
In this sense, the rows of the Hadamard matrix are the ``bang-bang'' modulation functions of the DD sequences derived from them.
The number of pulses each qubit of color $c$ undergoes per sequence repetition is determined by counting the number of sign changes per row $i=g(c)$, including the wrap-around sign change from $j=N-1$ to $j=0$.
All sequences in this work are single-axis, i.e., they involve only $X$ pulses.

The following result is a direct extension of the single-axis CHaDD theorem of Ref.~\cite{brown2024efficient}, replacing a minimum $\chi(G)$-coloring by an arbitrary proper $C$-coloring and adding bounds on the PRR.
Since these changes are central to the constructions below, for completeness and the reader's convenience we repeat the theorem here and its proof in Appendix~\ref{app:CHaDD-thm-proof}.

\begin{mytheorem}[Single-axis Chromatic-Hadamard DD (CHaDD)]
\label{thm:chadd}
A circuit of depth $C < N \leq 2C$,
with $\frac14 + \frac1{2C} \leq \PRR < \frac12 + \frac7{6C}$,
involving only $X$ pulses,
suffices to cancel to first order in $\tau$
all terms in $H_G$ excluding $H_1^x$ and $H_2^{xx}$
on a qubit connectivity graph $G$ that can be properly $C$-colored.
\end{mytheorem}
Note that unlike Ref.~\cite{brown2024efficient}, in addition to $C$ not needing to be equal to the chromatic number $\chi(G)$, here we adopt an unnormalized sign-matrix convention, make explicit the connection between the toggling-frame average and the physical cycle propagator, and include bounds on the PRR.
The PRR bounds hold for the row assignment that minimizes the total pulse count, derived in \cref{sec:cwdd}; a general row assignment achieves the same cancellation, generally at a higher PRR.

To illustrate CHaDD, consider the following $C=3$ example.
In this case, $\nu=2$ bits are needed to express every color $c=1,\dots,C$ in binary, and an $N=2^\nu=4$-dimensional Hadamard matrix suffices to decouple the graph Hamiltonian to first order in time, as shown in \cref{fig:C=3}.

\begin{figure}[h]
\begin{center}
\begin{tikzpicture}[
    matrixnode/.style={inner sep=2pt},
    arrow/.style={-{Latex[length=3mm]}, thick},
    every node/.append style={font=\large}
]

% The sign matrix
\node (matrix) [matrixnode] {
$\begin{matrix}
    + & + & + & + & \hspace{7mm} \texttt{IIII}\\
    + & - & + & - & \hspace{7mm} \texttt{XXXX} \\
    + & + & - & - & \hspace{7mm} \texttt{IXIX} \\
    + & - & - & + & \hspace{7mm} \texttt{XIXI}
\end{matrix}$};

% Circles and arrows
\node[circle, draw=none, fill=gray!60, minimum size=5mm, left=2cm of matrix.north west, yshift=-9pt] (c1) {};
\node[circle, draw=none, fill=blue!70, minimum size=5mm, left=2cm of matrix.west, yshift=7pt] (c2) {};
\node[circle, draw=none, fill=red!70, minimum size=5mm, left=2cm of matrix.west, yshift=-7pt] (c3) {};
\node[circle, draw=none, fill=green!70!black, minimum size=5mm, left=2cm of matrix.south west, yshift=9pt] (c4) {};

% Draw arrows
\draw[arrow] (c1.east) -- ++(0.5,0) |- ([xshift=-2pt,yshift=-9pt]matrix.north west);
\draw[arrow] (c2.east) -- ++(0.5,0) |- ([xshift=-2pt,yshift=7pt]matrix.west);
\draw[arrow] (c3.east) -- ++(0.5,0) |- ([xshift=-2pt,yshift=-7pt]matrix.west);
\draw[arrow] (c4.east) -- ++(0.5,0) |- ([xshift=-2pt,yshift=9pt]matrix.south west);

\end{tikzpicture}
\end{center}
\caption{Illustration of the color assignment, Hadamard matrix, and DD sequences for $C=3$.}
    \label{fig:C=3}
\end{figure}
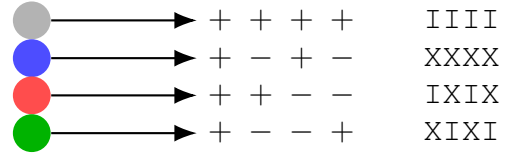

The CHaDD sequence designed for the active embedded $C=3$ coloring in \cref{fig:heavy-hex-embedded-3coloring}, and tested in \cref{sec:experiments}, uses the color-to-row map $i=g^*(c)$ given by $g^*(1)=2$, $g^*(2)=3$, and $g^*(3)=1$ (the row-optimized assignment introduced in \cref{sec:cwdd} and exemplified in \cref{fig:C=3}).
The spectator gray qubits, which are not counted among the $C$ colors, are assigned to the constant top row of the Hadamard matrix and therefore idle for the $N=4$ time steps, following the sequence \texttt{IIII}.
These sequences are read from left to right, with each character corresponding to a time step of duration $\tau$: \texttt{I} signifies free evolution $f_\tau$ (idling) for duration $\tau$, while \texttt{X} indicates that the qubit idles for $\tau-\delta$ and then undergoes an $X$ pulse of duration $\delta$.
These pulses are right-aligned so that in the ``bang-bang'' limit as $\delta\rightarrow0$, the pulse occurs exactly at the end of each time step, at an integer multiple of $\tau$, and the sign matrix rows can be interpreted as the values of a piecewise constant modulation function in each time step.
Row $i=2$ changes sign between columns $j$ and $(j+1) \bmod N$ in time steps $j=1,3$, so red qubits ($c=1$) follow the sequence \texttt{IXIX}.
Row $i=3$ changes sign in time steps $j=0,2$, so green qubits ($c=2$) follow \texttt{XIXI}, while row $i=1$ changes sign in every time step, so blue qubits ($c=3$) follow \texttt{XXXX}.

\section{Results}
\label{sec:results}

\Cref{thm:chadd} uses the minimum exponent $\nu_{\min}=\lceil\log_2(C+1)\rceil$.
More generally, any Hadamard matrix $W_\nu$ with $\nu\geq\nu_{\min}$ gives the same first-order decoupling, with circuit depth $N=2^\nu$; the depth is larger than necessary when $\nu>\nu_{\min}$.
We show below that this can be advantageous.
Hadamard matrices are not the only sign matrices capable of achieving this decoupling.
Only two conditions on the sign matrix are essential: each color must be assigned to a row that adds up to zero, so that $\overline{H}_1=H_1^x$, and qubits of different colors must be assigned to orthogonal rows, so that all but the $XX$ crosstalk between them is suppressed, whether or not the differently colored qubits are connected by an edge in the graph.

To formalize these observations, the map
\begin{align}
    \Phi:\mathbb{F}_2^\nu &\longrightarrow
    \left\langle \widetilde{X}_1,\dots,\widetilde{X}_C\right\rangle,
    &
    \Phi(j)&=U_j ,
\end{align}
where $\widetilde{X}_c$ and $U_j$ are respectively defined in \cref{eq:xtilde,eq:chadd-decoupling-unitary}, is a group homomorphism since
\begin{align}
    U_jU_k
    =
    \bigotimes_{c=1}^C
    \widetilde{X}_c^{g(c)\cdot(j\oplus k)}
    =
    U_{j\oplus k},
\end{align}
where $\oplus$ denotes bitwise addition modulo 2.
Let $r=\dim\operatorname{span}\{g(c):1\leq c\leq C\}$ over $\mathbb{F}_2$.
Then $\ker\Phi$ is the orthogonal complement of this span, and the group of distinct control unitaries is
\begin{align}
    \mathcal{G}=\operatorname{Im}\Phi
    \cong \mathbb{F}_2^\nu/\ker\Phi
    \cong \mathbb{F}_2^r .
\end{align}
Consequently, the $N=2^\nu$ columns list each distinct control unitary $2^{\nu-r}$ times.
For the minimum choice $\nu=\lceil\log_2(C+1)\rceil$, the selected row indices necessarily span $\mathbb{F}_2^\nu$: otherwise they would lie in a proper subspace, which contains at most $2^{\nu-1}-1<C$ nonzero vectors.
Thus $r=\nu$ and $|\mathcal{G}|=N$ for minimum-depth CHaDD, whereas a larger Hadamard matrix can yield a non-trivial kernel and repeated control unitaries.
The pulse at the end of time step $j$ is $U_{(j+1)\bmod N}U_j^\dagger$.
Permuting the columns merely reorders the same toggling Hamiltonians and therefore preserves the first-order average (provided the all-$+$ column remains first, so that no pulse precedes the first time step), while generally changing the number and timing of pulses.

These observations motivate the new sequences we present in this section, which share the first-order average with CHaDD, but, for $C>2$, start from larger Hadamard matrices.

\subsection{Chromatic-Binary DD (CBDD)}
\label{sec:cbdd}

As noted above, any sign matrix whose color-assigned rows add up to zero and are pairwise orthogonal suffices to schedule DD sequences for decoherence and crosstalk suppression.
A $(C+1) \times 2^C$ binary matrix
\begin{align}
\label{eq:binary-matrix}
    B_C
    \equiv
    \sum_{i=0}^C \sum_{j=0}^{2^C-1} (-1)^{ j \cdot 2^{C-i} } \ketb{i}{j}
\end{align}
is one such example, where each column $j=0,\dots,2^C-1$ represents $j$ as a $(C+1)$-bit string of signs $+1$ and $-1$, corresponding to 0s and 1s respectively in the bitstring representation of $j$, with the least significant bit in row $i=C$ and the most significant (leading zero) bit in row $i=0$.
This convention yields a top row and leftmost column of all $+$'s, as for the Hadamard matrix, e.g., for $C=3$,
\begin{align}
\label{eq:ex-binary-matrix}
    B_3
    =
    \begin{bmatrix}
        + & + & + & + & + & + & + & + \\
        + & + & + & + & - & - & - & - \\
        + & + & - & - & + & + & - & - \\
        + & - & + & - & + & - & + & - \\
    \end{bmatrix}
    .
\end{align}
To explain the $(i,j)$th entry of this matrix, $(-1)^{j \cdot 2^{C-i}}$,
note that the binary representation of $2^{C-i}$ has a single $1$, in position $C-i$ (counting from the least significant bit in position $0$), so the dot product $j \cdot 2^{C-i}$ extracts the bit of $j$ in that position; for the nonconstant rows $i=1,\dots,C$, these binary representations are the $C$ standard basis vectors of $\mathbb{F}_2^C$.
E.g., for $C=3$, column $j=5$ is represented in $C+1=4$ bits, with a leading zero, as $j=0101$, and translated into signs as $[+-+-]^T=B_3\ket{5}$. 

For each nonconstant row $i=1,\dots,C$ of $B_C$,
\begin{align}
    \bra{i} B_C
    =
    \sum_{j=0}^{2^C-1} (-1)^{ j \cdot 2^{C-i} } \bra{j}
    =
    \bra{2^{C-i}} W_C .
\end{align}
The constant row satisfies $\bra{0}B_C=\bra{0}W_C$.
Thus $B_C$ is a row submatrix of the dimension-$2^C$ Hadamard matrix $W_C$, which is larger than the minimum matrix $W_\nu$, with $\nu=\lfloor\log_2 C\rfloor+1$, whenever $C>2$.
The resulting circuit depth $N=2^C$ scales exponentially with the number of colors rather than linearly.
Nevertheless, the averaging argument of \cref{thm:chadd} applies unchanged to the $C$ nonconstant rows, so Chromatic-Binary DD (CBDD) yields $\overline{H}_G=H_1^x+H_2^{xx}$.

\begin{mytheorem}[Single-axis Chromatic-Binary DD (CBDD)]
\label{thm:cbdd}
A circuit of depth $N=2^C$ with $\PRR<2/C$, involving only $X$ pulses, suffices to cancel to first order in $\tau$ 
all terms in $H_G$ excluding $H_1^x $ and $H_2^{xx}$
on a qubit connectivity graph $G$ that can be properly $C$-colored.
\end{mytheorem}

\begin{proof}
As before, let $G=(V,E)$ be a graph representing a 2-local Hamiltonian $H_G$ to be decoupled, and let $f$ be a proper $C$-coloring of $G$.
Assign the colors bijectively to the $C$ nonconstant rows of $B_C$ via a permutation $g:\{1,\dots,C\}\to\{1,\dots,C\}$, with $i=g(c)$.
Because every assigned row is a distinct nonconstant row of $W_C$, the averaging argument of \cref{thm:chadd} gives $\overline{H}_G=H_1^x+H_2^{xx}$.
The circuit depth is $N=2^C$.

An $X$ pulse is applied to qubits of color $c$ in time step $j$ whenever there is a sign change in row $g(c)$ between columns $j$ and $(j+1) \bmod N$. 
Since the least significant bit of a binary integer changes with every increment, qubits assigned to the bottom row of the binary matrix, $i=C$, undergo an $X$ pulse during every time step.
The next least significant bit changes every two time steps, so qubits assigned to the second from the bottom row of the binary matrix, $i=C-1$, undergo an $X$ pulse once every two time steps, and so on up to the most significant non-leading-zero bit, corresponding to row $i=1$ of the binary matrix, which changes sign once halfway through the sequence and once more, restoratively, at the wrap-around, resulting in two $X$ pulses.
Thus, qubits of color $c$ assigned to row $i=g(c)$ undergo $P_c=2^i$ pulses per CBDD sequence repetition.
This results in a total of $P=\sum_{i=1}^C 2^i = 2^{C+1}-2$ pulses across $C$ colors and $N=2^C$ time steps, for
a PRR of
\begin{align}
\label{eq:cbdd-prr}
    \frac{2^{C+1}-2}{C ~ 2^C} < \frac2C.
\end{align}
\end{proof}

The CBDD control unitaries are
\begin{align}
    \label{eq:cbdd-group}
    U_j = \bigotimes_{c=1}^C
    \widetilde{X}_c^{j \cdot 2^{C-g(c)}} .
\end{align}
Because $g$ is a permutation, the binary representations of the row indices $2^{C-g(c)}$ run over all $C$ standard basis vectors of $\mathbb{F}_2^C$.
Hence the map $j\mapsto U_j$ is injective,
$U_jU_k=U_{j\oplus k}$, and the CBDD control group has size $2^C$ and is isomorphic to the additive group of the vector space $\mathbb{F}_2^C$.

Note that because the binary representations of the row indices $2^{C-g(c)}$ of $B_C$ are the $C$ standard basis vectors of $\mathbb{F}_2^C$, and are thus linearly independent, the entrywise product of the rows assigned to any nonempty set of colors is again a nonconstant row of $W_C$.
Consequently, in terms of the $K$-local Hamiltonian of \cref{eq:k-local-hamiltonian}, CBDD cancels to first order every $k$-body term whose $Y$- and $Z$-type factors act on an odd number of qubits of at least one color, not only the two-body terms of $H_G$.
CHaDD row assignments need not share this property: e.g., the entrywise product of rows 1, 2, and 3 of $W_\nu$ is the constant row, so a $YYY$ term acting on three differently colored qubits assigned to these rows survives the $C=3$ CHaDD average.
A systematic treatment of such selective $k$-local cancellations, based on classical codes, was recently given in Ref.~\cite{Nguyen2026color}.

Note that $C=2$, in addition to being too small to be illustrative, produces a sign matrix whose rows are equal to rows $2$ and $1$ of the Hadamard matrix featured in \cref{fig:C=3}. Thus, 
we illustrate how to extract the CBDD pulse sequence from a sign matrix with a $C=3$ example; see \cref{fig:C=3-CBDD}.

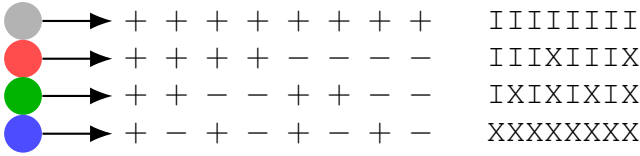
\begin{figure}[h]
\begin{center}
\begin{tikzpicture}[
    matrixnode/.style={inner sep=2pt},
    arrow/.style={-{Latex[length=3mm]}, thick},
    every node/.append style={font=\large}
]

% The sign matrix
\node (matrix) [matrixnode] {
$\begin{matrix}
    + & + & + & + & + & + & + & + & \hspace{5mm} \texttt{IIIIIIII} \\
    + & + & + & + & - & - & - & - & \hspace{5mm} \texttt{IIIXIIIX} \\
    + & + & - & - & + & + & - & - & \hspace{5mm} \texttt{IXIXIXIX} \\
    + & - & + & - & + & - & + & - & \hspace{5mm} \texttt{XXXXXXXX} \\
\end{matrix}$};

% Circles and arrows
\node[circle, draw=none, fill=gray!60, minimum size=5mm, left=1cm of matrix.north west, yshift=-9pt] (c1) {};
\node[circle, draw=none, fill=red!70, minimum size=5mm, left=1cm of matrix.west, yshift=7pt] (c2) {};
\node[circle, draw=none, fill=green!70!black, minimum size=5mm, left=1cm of matrix.west, yshift=-7pt] (c3) {};
\node[circle, draw=none, fill=blue!70, minimum size=5mm, left=1cm of matrix.south west, yshift=9pt] (c4) {};

% Draw arrows
\draw[arrow] (c1.east) -- ++(0.5,0) |- ([xshift=-2pt,yshift=-9pt]matrix.north west);
\draw[arrow] (c2.east) -- ++(0.5,0) |- ([xshift=-2pt,yshift=7pt]matrix.west);
\draw[arrow] (c3.east) -- ++(0.5,0) |- ([xshift=-2pt,yshift=-7pt]matrix.west);
\draw[arrow] (c4.east) -- ++(0.5,0) |- ([xshift=-2pt,yshift=9pt]matrix.south west);

\end{tikzpicture}
\end{center}
\caption{Illustration of the color assignment, binary matrix $B_3$, and CBDD sequences for $C=3$, with the identity color-to-row map $i=g(c)=c$.}
    \label{fig:C=3-CBDD}
\end{figure}

To translate the example sign matrix for $C=3$ in \cref{eq:ex-binary-matrix} into a sequence for the active embedded 3-coloring in \cref{fig:heavy-hex-embedded-3coloring}, we take $i=g(c)=c$ and assign the spectator gray qubits, which are not counted among the $C$ colors, to the constant row $i=0$.
Since the top row consists of only $+$'s, the gray qubits do not undergo any pulses and idle for the duration of the sequence, $N=2^C=8$ time steps, so they follow the schedule \texttt{IIIIIIII}.
Row $i=1$ changes sign between columns $j$ and $(j+1)\bmod N$ only in time steps $j=3,7$, so red qubits ($c=1$) follow the schedule \texttt{IIIXIIIX}.
Row $i=2$ changes sign in time steps $j=1,3,5,7$, so green qubits ($c=2$) follow \texttt{IXIXIXIX}, while row $i=3$ changes sign in every time step, so blue qubits ($c=3$) follow \texttt{XXXXXXXX}.

\begin{table*}
\centering
\begin{tabular}{ |c|c|c|c| } 
 \hline
 DD sequence & $P$ & $N$ & PRR ($P/NC$) \\ \hline
 CWDD & $\frac{C(C+1)}2 + \lceil \frac{C}2 \rceil$ & $N=2^{\lfloor \log_2 C \rfloor + 1} \in \mathcal{O}(C)$, $C < N \leq 2C$ & $\frac14 + \frac1{2C} \leq \PRR < \frac12 + \frac7{6C}$ \\ 
 CBDD & $2^{C+1}-2$ & $N=2^C$ & $\frac{2^{C+1}-2}{C 2^C} < \frac2C$ \\
 CGDD & $2^C$ & $N=2^C$ & $\frac1C$ \\
 \hline
\end{tabular}
\caption{
The number of pulses per sequence repetition $P$, circuit depth $N$, and pulse repetition rate $\PRR=P/(NC)$ of CWDD (row-optimized linear-depth CHaDD at fixed Hadamard column order; \cref{sec:cwdd}), CBDD, and CGDD.
}
\label{tab:prr-table}
\end{table*}

\subsection{Chromatic-Gray DD (CGDD)}
\label{sec:cgdd}

For $C\geq2$, we can achieve the same first-order average and circuit depth as CBDD with a lower pulse count by rearranging the columns of the binary matrix to minimize the number of sign changes between consecutive columns.
The reflected binary code, or ``Gray code,'' permutation accomplishes exactly this.
Indeed, every cyclic ordering of the $2^C$ distinct columns incurs at least one sign change per transition, so $P\geq2^C$ for any such reordering; the Gray order attains this minimum, since consecutive Gray codewords differ in exactly one bit.
Thus CGDD minimizes the pulse count among sequences that traverse the full control group once per cycle.
At the same depth, sequences that traverse a smaller control group with repeated elements can have a lower pulse count still (e.g., the four lowest-sequency nonconstant rows of $W_4$ give $P=12<16$ at $C=4$); the dilated $C=3$ sequence of \cref{fig:chadd-equivalence-classes}(a) is of this repeated-element type.
The $C$-bit Gray code for integer $j$ is $\gamma_C(j) = j \oplus \lfloor j/2 \rfloor$, where $\lfloor j/2 \rfloor$ is $j$ with its least significant bit deleted.

The $(C+1) \times 2^C$ Gray matrix $G_C$ is then constructed from the binary matrix $B_C$ [\cref{eq:binary-matrix}] by permuting the columns according to $j \mapsto \gamma_C(j)$:
\begin{align}
\label{eq:gray-matrix}
    G_C
    \equiv
    \sum_{i=0}^C \sum_{j=0}^{2^C-1} (-1)^{ \gamma_C(j) \cdot 2^{C-i} } \ketb{i}{j}
    .
\end{align}
The $j$th column is left-padded with a single zero, as with the binary matrix [\cref{eq:binary-matrix}], so that the top row and leftmost column consist of all $+$'s, as for the Hadamard matrix [\cref{eq:hadamard-matrix}];
e.g., for $C=3$, the Gray matrix is
\begin{align}
    \label{eq:ex-gray-matrix}
    G_3
    =
    \begin{bmatrix}
        + & + & + & + & + & + & + & + \\
        + & + & + & + & - & - & - & - \\
        + & + & - & - & - & - & + & + \\
        + & - & - & + & + & - & - & + \\
    \end{bmatrix}
    .
\end{align}
The sum of the entries of each row and the dot product of any two rows are invariant under permutation of the columns, so CGDD inherits the first-order decoupling average and circuit depth of CBDD; in addition, adjacent columns of $G_C$ have Hamming distance 1, so pulses of only a single color are applied in each time step.

\begin{mytheorem}[Single-axis Chromatic-Gray DD (CGDD)]
\label{thm:cgdd}
A circuit of depth $N=2^C$ with $\PRR=1/C$, involving only $X$ pulses, suffices to cancel to first order in $\tau$ 
all terms in $H_G$ excluding $H_1^x $ and $H_2^{xx}$
on a qubit connectivity graph $G$ that can be properly $C$-colored.
\end{mytheorem}

\begin{proof}
Again, let $G=(V,E)$ be a graph representing $H_G$, and let $f$ be a proper $C$-coloring of $G$.
Assign the colors bijectively to the $C$ nonconstant rows of $G_C$ via a permutation $g:\{1,\dots,C\}\to\{1,\dots,C\}$.
Since the columns of $G_C$ are a permutation of those of $B_C$, the row sums and pairwise inner products are unchanged; hence, by \cref{thm:cbdd}, $\overline{H}_G=H_1^x+H_2^{xx}$ and $N=2^C$.

The reflected Gray code is cyclic:
$\gamma_C(j)$ and $\gamma_C\left[(j+1)\bmod 2^C\right]$ differ in exactly one bit for every $j=0,\dots,2^C-1$, including the wrap-around transition.
Because $g$ permutes all $C$ nonconstant rows, the qubits of exactly one color undergo an $X$ pulse at each of the $N=2^C$ transitions.
Thus $P=2^C$ and $\PRR=P/(NC)=1/C$, which is $1/(2-2^{1-C})$ times the CBDD value of \cref{eq:cbdd-prr}.
\end{proof}

CGDD traverses the same set of control unitaries as CBDD, but in Gray-code order:
\begin{align}
    \label{eq:cgdd-group}
    U_j = \bigotimes_{c=1}^C
    \widetilde{X}_c^{\gamma_C(j)\cdot 2^{C-g(c)}} .
\end{align}
The Gray map $\gamma_C:j\mapsto j\oplus\lfloor j/2\rfloor$ is an invertible linear map on $\mathbb{F}_2^C$.
It therefore relabels the CBDD group by an automorphism:
$U_jU_k=U_{j\oplus k}$.
Consequently, the CBDD and CGDD control groups are the same set and are both isomorphic to the additive group of $\mathbb{F}_2^C$; only their traversal order, and hence their pulse counts, differ.
Moreover, since $\gamma_C$ is linear, each nonconstant row of $G_C$ is itself a row of $W_C$:
row $i$ of $G_C$ is equal to row $R_C[\gamma_C(2^{i-1})]$ of $W_C$, where $R_C$ is bit reversal on $C$ bits, e.g., rows 1, 2, and 3 of \cref{eq:ex-gray-matrix} are equal to rows 4, 6, and 3 of $W_3$.
Thus CGDD, like CBDD, is a CHaDD row assignment at depth $2^C$, and it has the same $k$-local cancellations, since column permutations preserve entrywise row products.

We demonstrate how to extract the CGDD sequence for $C=3$; see \cref{fig:C=3-CGDD}. Again, $C=2$ is unsuitable for an example since the rows of that Gray matrix correspond to rows $2$ and $3$ of the Hadamard matrix in \cref{fig:C=3}.

\begin{figure}
\begin{center}
\begin{tikzpicture}[
    matrixnode/.style={inner sep=2pt},
    arrow/.style={-{Latex[length=3mm]}, thick},
    every node/.append style={font=\large}
]

% The sign matrix
\node (matrix) [matrixnode] {
$\begin{matrix}
    + & + & + & + & + & + & + & + & \hspace{5mm} \texttt{IIIIIIII} \\
    + & + & + & + & - & - & - & - & \hspace{5mm} \texttt{IIIXIIIX} \\
    + & + & - & - & - & - & + & + & \hspace{5mm} \texttt{IXIIIXII} \\
    + & - & - & + & + & - & - & + & \hspace{5mm} \texttt{XIXIXIXI} \\
\end{matrix}$};

% Circles and arrows
\node[circle, draw=none, fill=gray!60, minimum size=5mm, left=1cm of matrix.north west, yshift=-11pt] (c1) {};
\node[circle, draw=none, fill=red!70, minimum size=5mm, left=1cm of matrix.west, yshift=7pt] (c2) {};
\node[circle, draw=none, fill=green!70!black, minimum size=5mm, left=1cm of matrix.west, yshift=-9pt] (c3) {};
\node[circle, draw=none, fill=blue!70, minimum size=5mm, left=1cm of matrix.south west, yshift=9pt] (c4) {};

% Draw arrows
\draw[arrow] (c1.east) -- ++(0.5,0) |- ([xshift=-2pt,yshift=-11pt]matrix.north west);
\draw[arrow] (c2.east) -- ++(0.5,0) |- ([xshift=-2pt,yshift=7pt]matrix.west);
\draw[arrow] (c3.east) -- ++(0.5,0) |- ([xshift=-2pt,yshift=-9pt]matrix.west);
\draw[arrow] (c4.east) -- ++(0.5,0) |- ([xshift=-2pt,yshift=9pt]matrix.south west);

\end{tikzpicture}
\end{center}
\caption{Illustration of the color assignment, Gray matrix $G_3$, and CGDD sequences for $C=3$, with the identity color-to-row map $i=g(c)=c$.}
    \label{fig:C=3-CGDD}
\end{figure}
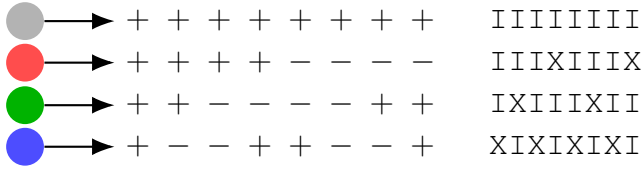

We can read off the $C=3$ CGDD sequence from \cref{eq:ex-gray-matrix} using the identity color-to-row map $i=g(c)=c$.
The spectator gray qubits in \cref{fig:heavy-hex-embedded-3coloring}, which are not counted among the $C$ colors, are assigned to the constant row $i=0$ and therefore follow \texttt{IIIIIIII} for the $N=2^C=8$ time steps.
Row $i=1$ changes sign between columns $j$ and $(j+1)\bmod N$ in time steps $j=3,7$, so red qubits ($c=1$) follow the schedule \texttt{IIIXIIIX}.
Row $i=2$ changes sign in time steps $j=1,5$, so green qubits ($c=2$) follow \texttt{IXIIIXII}, while row $i=3$ changes sign in time steps $j=0,2,4,6$, so blue qubits ($c=3$) follow \texttt{XIXIXIXI}.

\begin{figure}[t]
\centering
\includegraphics[width=0.98\columnwidth]{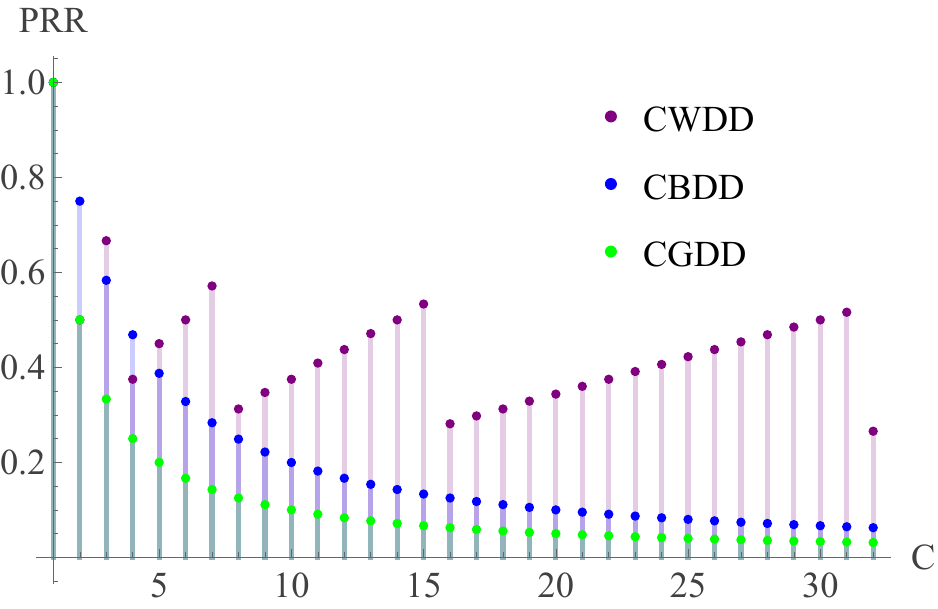}
\caption{
PRR as a function of the number of colors $C$ for CWDD, i.e., row-optimized linear-depth CHaDD at fixed Hadamard column order (purple dots), CBDD (blue dots), and CGDD (green dots).
Except at small $C$ (all three coincide at $C=1$; CGDD and CWDD have the same PRR for $C=2$; and CWDD has lower PRR than CBDD for $C=2$ and $C=4$), CWDD has substantially greater PRR than CBDD or CGDD, while CBDD's PRR is $(2-2^{1-C})$ times CGDD's.
The jumps in the CWDD curve occur because its circuit depth $N=2^{\lfloor\log_2C\rfloor+1}$ changes discontinuously with $C$.
}
\label{fig:PRR-comparison}
\end{figure}

\subsection{Chromatic-Walsh DD (CWDD)}
\label{sec:cwdd}

At fixed Hadamard column order, the PRR of a linear-depth CHaDD sequence depends on which Hadamard rows are assigned to the colors.
A Walsh matrix is a row permutation of the Hadamard matrix ordered by \textit{sequency} $s$, the number of sign changes between adjacent columns, excluding the cyclic wrap-around.
Each Walsh matrix row $s$ has sequency $s$ and begins with $+1$, so its final sign differs from its initial sign exactly when $s$ is odd; its cyclic pulse count is therefore
\begin{align}
    P_s=s+(s\bmod 2).
\end{align}
Consequently, among row assignments at this fixed column order, the minimum total pulse count is obtained by choosing the $C$ nonconstant Walsh rows of lowest sequency, $s=1,\dots,C$.
The corresponding Hadamard row index is
$g^*(c)=R_\nu\left[\gamma_\nu(c)\right]$ for $c=1,\dots,C$, where $\gamma_\nu$ is the $\nu$-bit Gray map and $R_\nu$ is bit reversal on $\nu$ bits.
We refer to this row-optimized schedule as chromatic-Walsh DD (CWDD).
Walsh modulation has previously been used to construct and analyze single-qubit DD and error-suppressing control sequences~\cite{Hayes:2011aa,Ball:2015aa,Qi:2017aa}.
When $C=2^\nu-1$ (e.g., $C=3$), every bijection uses all $N-1$ nonconstant rows, so the total pulse count, and hence the PRR, is independent of the row assignment; only the per-color pulse counts differ.

The CWDD total pulse count is
\begin{align}
\label{eq:prr-optimal-chadd-pulse-count}
    P
    =\sum_{c=1}^C\left[c+(c\bmod2)\right]
    =\frac{C(C+1)}2+\left\lceil\frac C2\right\rceil ,
\end{align}
and hence
\begin{align}
    \frac C2+1
    \leq \frac PC
    =\frac{C+1}{2}+\frac1C\left\lceil\frac C2\right\rceil
    \leq \frac C2+\frac76 ,
\end{align}
where $\frac12\leq C^{-1}\lceil C/2\rceil\leq\frac23$ for $C\geq2$; for $C=1$, the corresponding PRR is $1$ and satisfies the bounds below directly.
Using $C<N\leq2C$, the CWDD PRR obeys
\begin{align}
\label{eq:prr-optimal-chadd}
    \frac14+\frac1{2C}
    \leq \PRR
    < \frac12+\frac7{6C},
\end{align}
with equality on the left iff $C\geq2$ is a power of 2 (whence $N=2C$ and $\lceil C/2\rceil=C/2$).
These are the PRR bounds stated in \cref{thm:chadd}.
The bounds optimize only the row assignment at the stated fixed column order; allowing arbitrary column permutations is a separate scheduling problem.

The pulse count $P$, circuit depth $N$, and PRR of CWDD, CBDD, and CGDD are summarized in \cref{tab:prr-table} and plotted in \cref{fig:PRR-comparison}.

\begin{figure*}[t]
\centering
\includegraphics[width=0.37\textwidth]{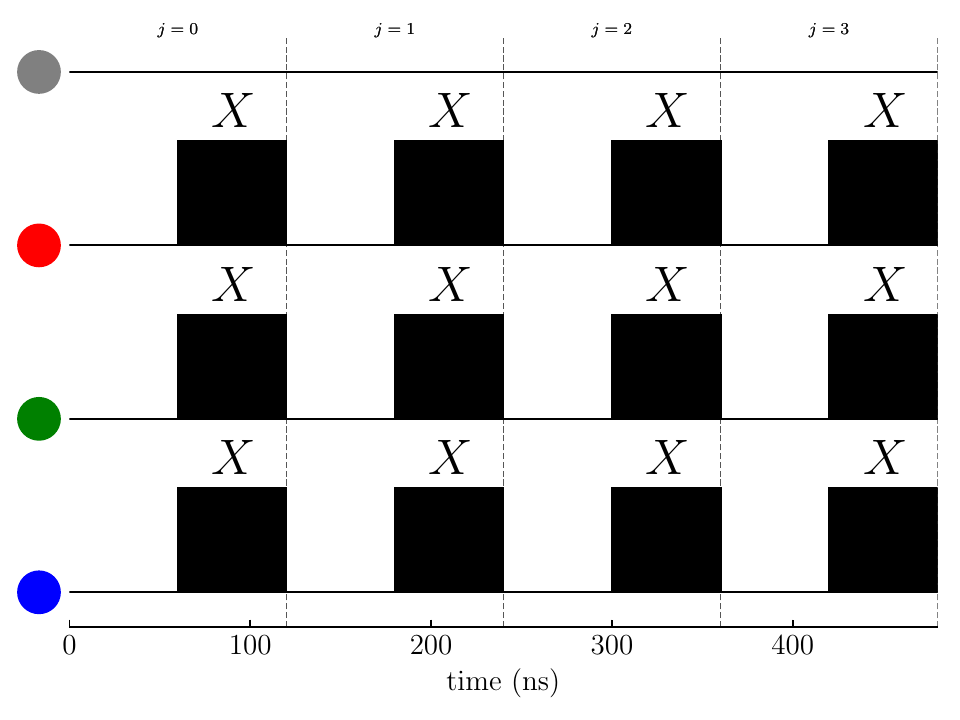}
\includegraphics[width=0.37\textwidth]{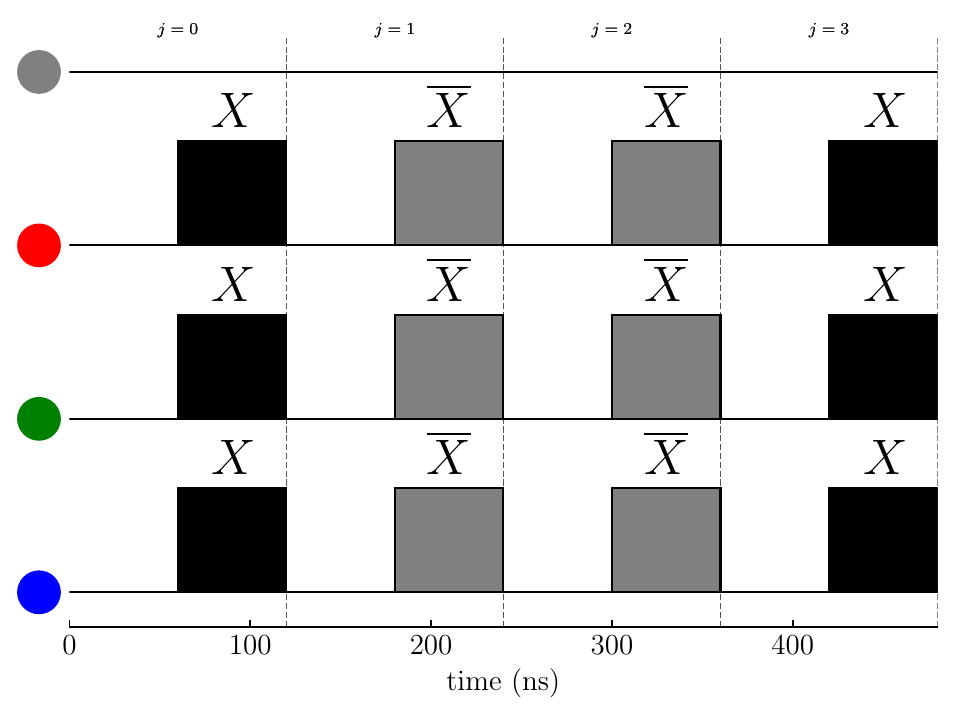}
\\
\includegraphics[width=0.37\textwidth]{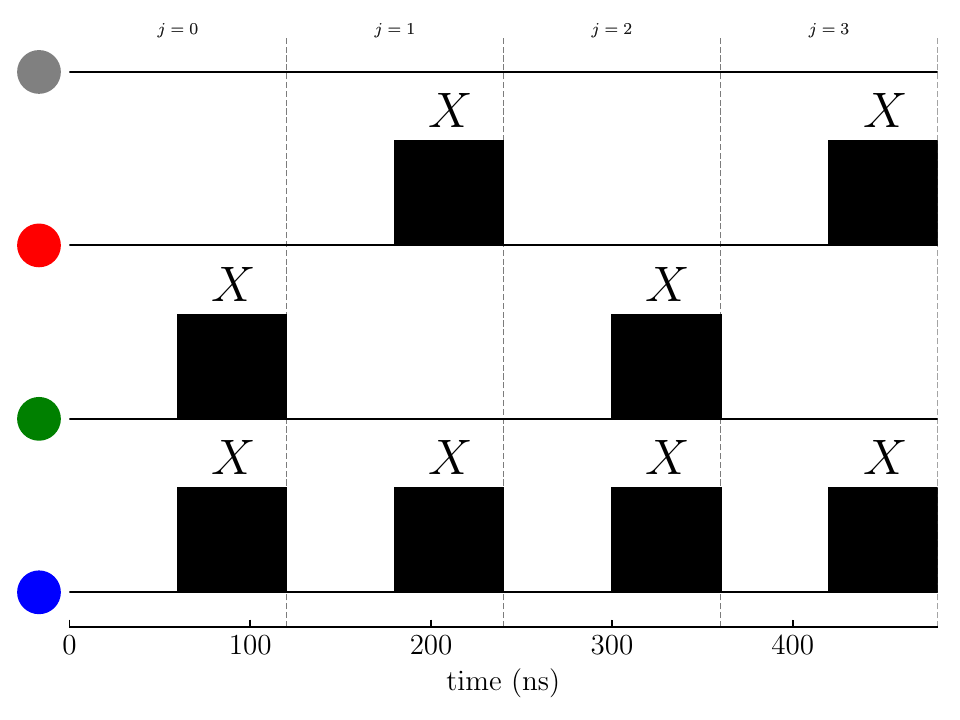}
\includegraphics[width=0.37\textwidth]{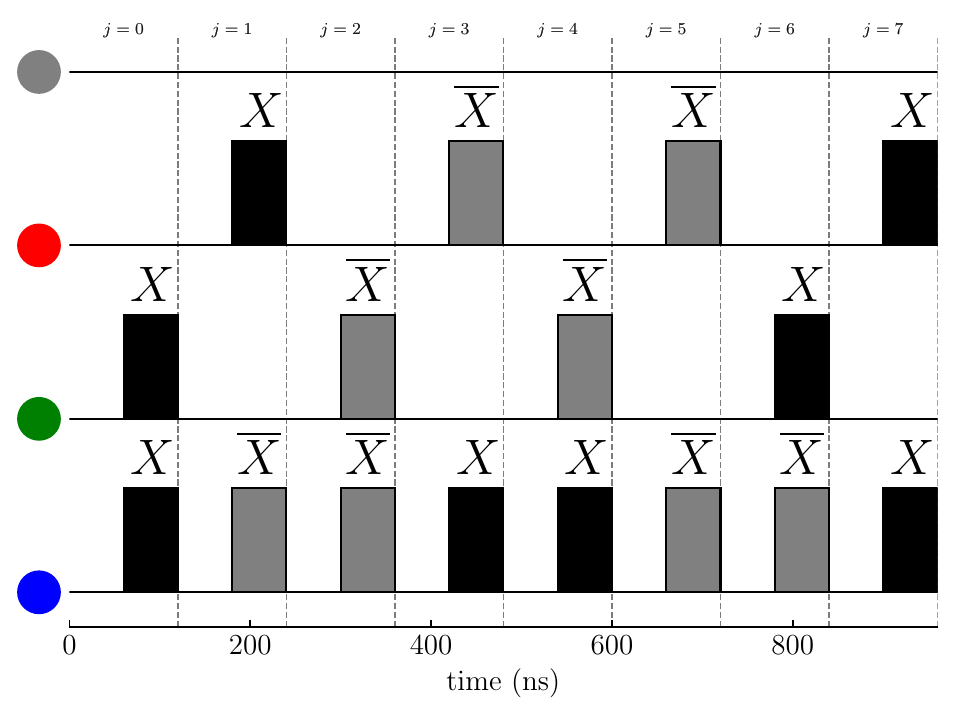}
\\
\includegraphics[width=0.37\textwidth]{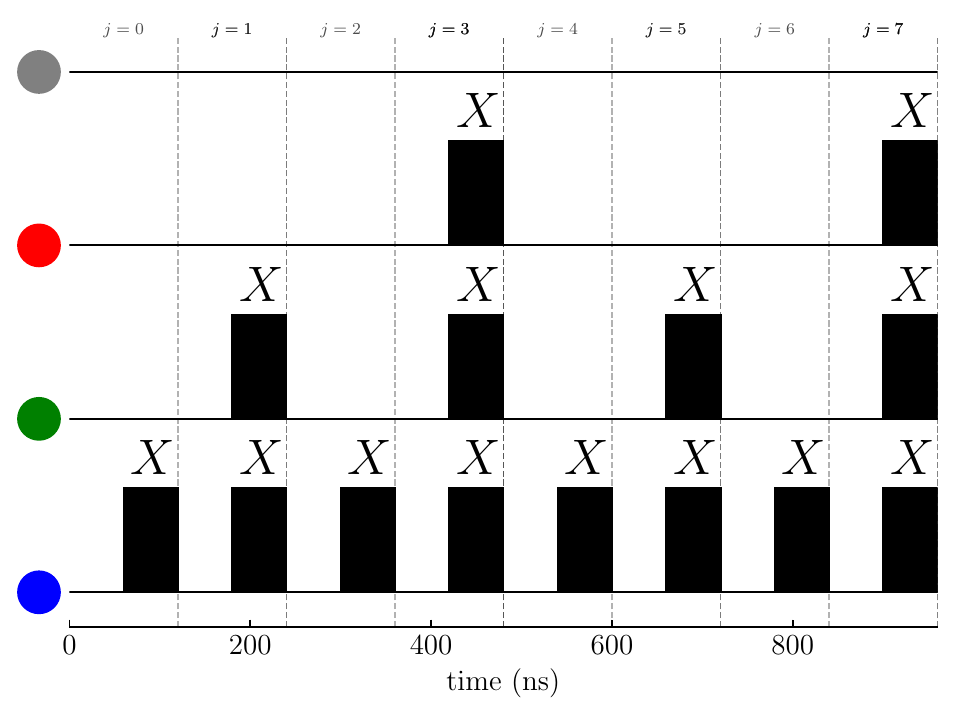}
\includegraphics[width=0.37\textwidth]{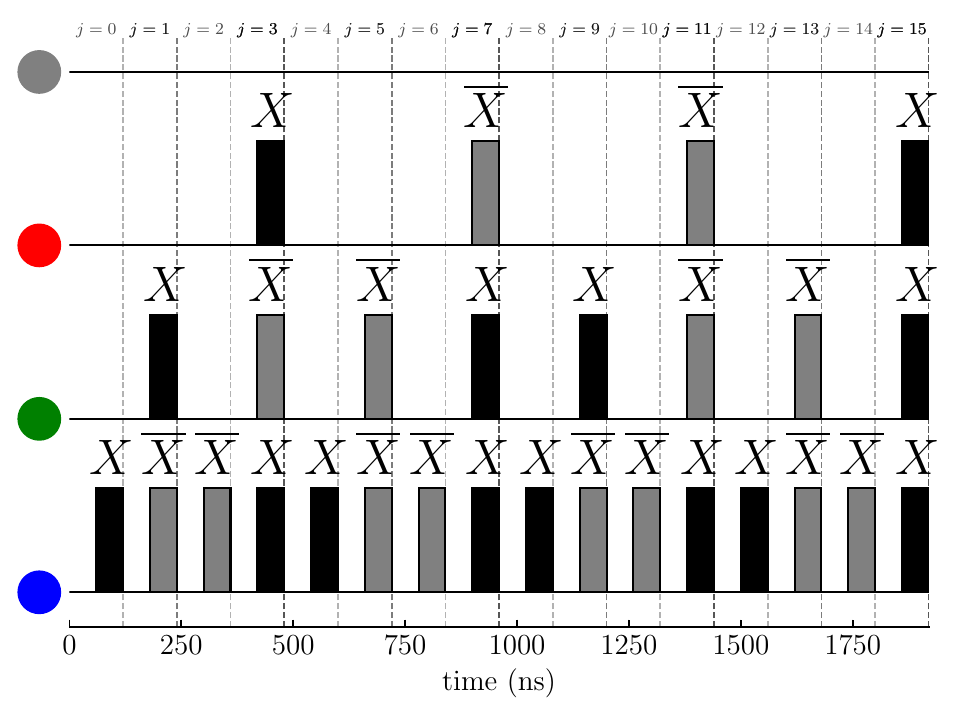}
\\
\includegraphics[width=0.37\textwidth]{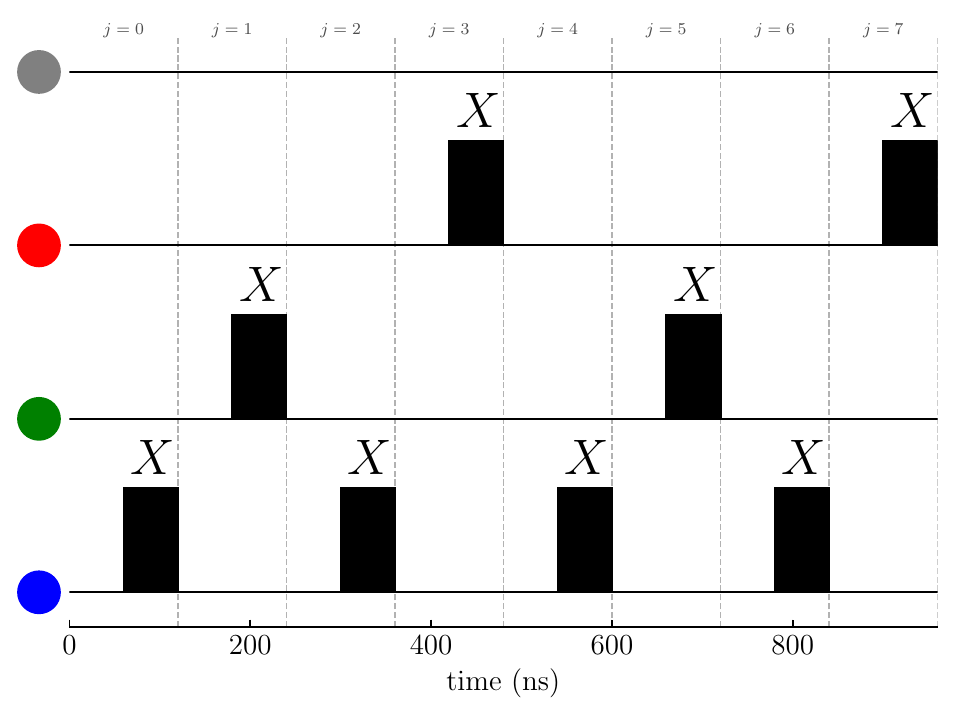}
\includegraphics[width=0.37\textwidth]{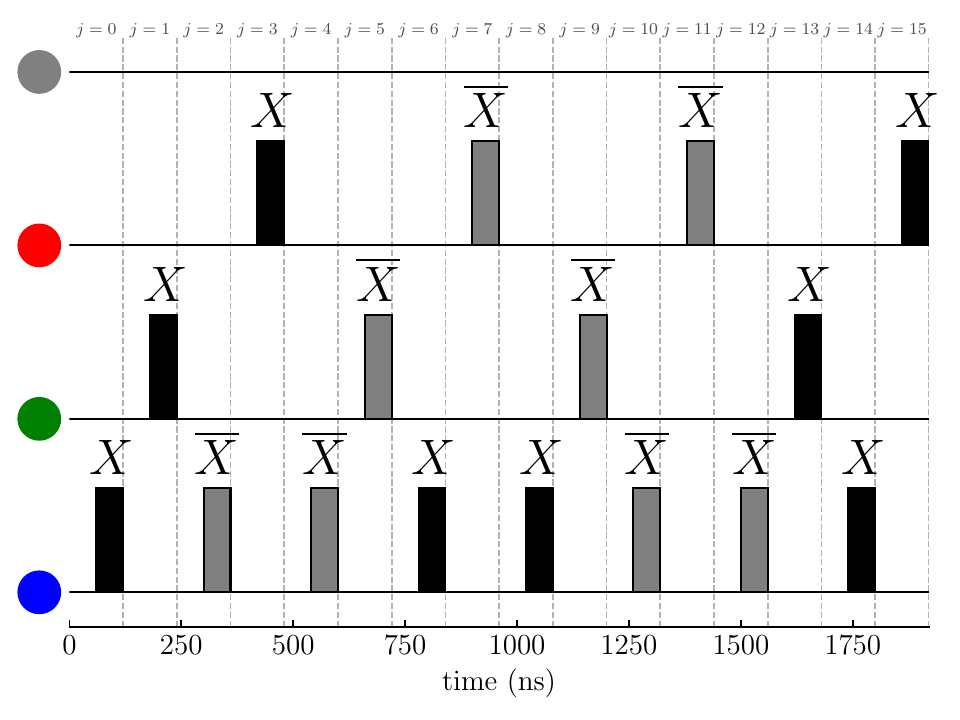}
\caption{
Timeline diagrams for embedded $C=3$ chromatic DD sequences with pulse width $\delta=60\,\mathrm{ns}$ and time interval $\tau=2\delta=120\,\mathrm{ns}$.
Multiples of $\tau$ are marked with vertical gray dashed lines and annotated above by the time step $j=0,1,\dots,N-1$ they mark the end of.
Left column, top to bottom: XX, the pure-$X$ DD sequence on RGB qubits and idle on gray qubits;
CHaDD, CBDD, and CGDD.
Right column, top to bottom: UR4, CHaDD-R, CBDD-R, CGDD-R,
the robust versions of their left-column counterparts,
in which, within each block of four $X$ pulses, the second and third are replaced with $\overline{X}$ pulses, i.e., $\pi$ pulses of opposite phase; the resulting four-pulse block $X\,\overline{X}\,\overline{X}\,X$ is the universally robust sequence UR4~\cite{Genov2017}.
The sequence is repeated where necessary to obtain a multiple of four pulses on each color, as is needed for every color with only two pulses per repetition (red and green in CHaDD and CGDD, and red in CBDD), in which case the robust chromatic sequences have twice the depth of their non-robust counterparts.
}
\label{fig:bivalent-3coloring-timeline-diagrams}
\end{figure*}

\section{Experiments}
\label{sec:experiments}

\begin{figure*}[t]
(a)
\includegraphics[width=0.98\columnwidth]{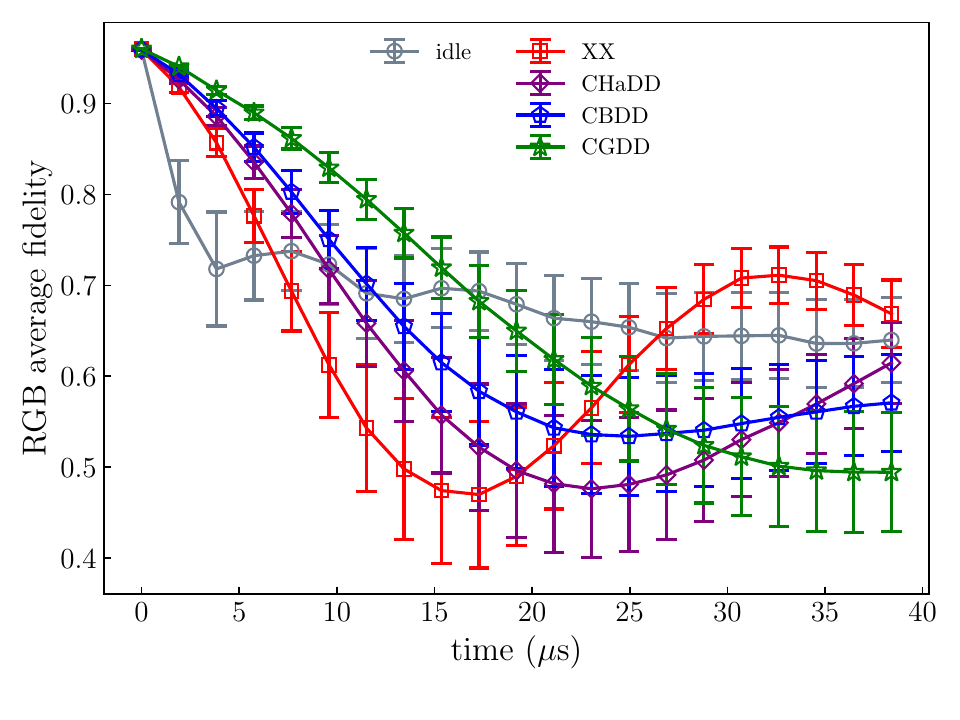}
(b)
\includegraphics[width=0.98\columnwidth]{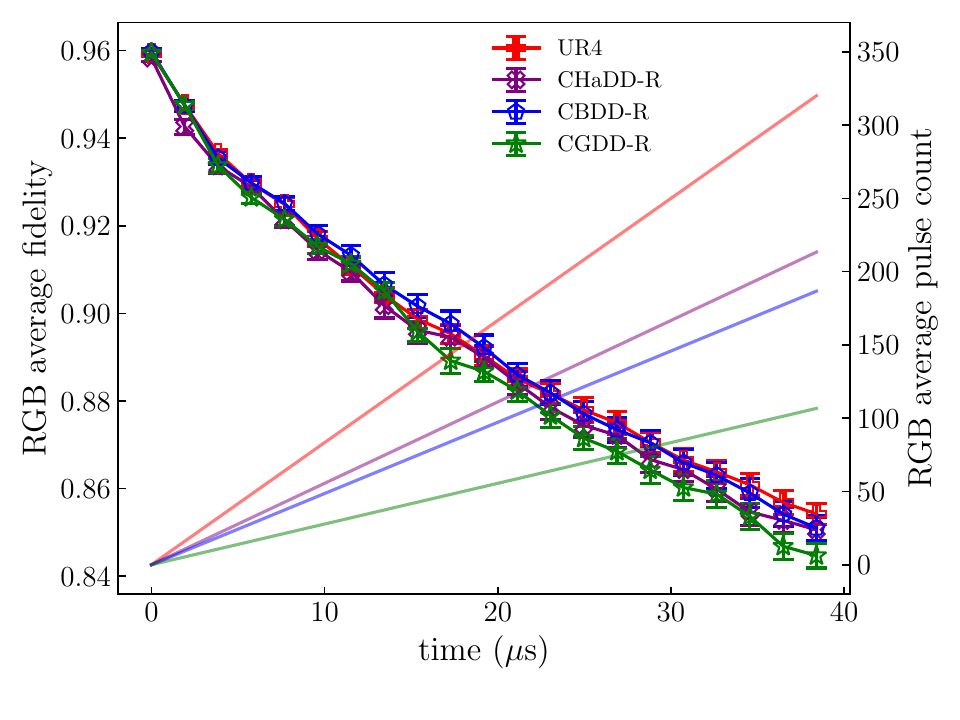}
\caption{
Robustness largely eliminates the PRR advantage.
(a)
Average fidelity of RGB qubits protected by idle evolution, XX, CHaDD, CBDD, and CGDD.
At short times, each successive sequence in the legend preserves the initial states better than the one before it.
The ordering of the non-idle sequences in the legend mirrors the layout of the timeline diagrams in the left column of \cref{fig:bivalent-3coloring-timeline-diagrams}.
(b)
Average fidelity of RGB qubits protected by the robust counterparts of the non-idle sequences at left, depicted as timeline diagrams in the right column of \cref{fig:bivalent-3coloring-timeline-diagrams}, plotted on the left vertical axis and average pulse count of RGB qubits on the right vertical axis, whose slope is RGB PRR$/\tau$.
The difference in performance observed in (a) has been largely eliminated in (b).
Error bars here and throughout are $1\sigma$ in each direction, where $\sigma$ is the standard error of the mean over the $17$ preparation settings, as described in \cref{sec:setup}.
}
\label{fig:robustness-eliminates-prr-advantage}
\end{figure*}

\subsection{Experimental Setup}
\label{sec:setup}

We performed state preservation experiments on the \texttt{ibm\_strasbourg} 127-qubit QPU in the IBM Eagle r3 family, which has single-qubit pulse duration $\delta=60\,\mathrm{ns}$, as in the sequence timeline diagrams in \cref{fig:bivalent-3coloring-timeline-diagrams}.
This family of QPUs lacks tunable couplers~\cite{Sete:2021aa}, which mitigate crosstalk, and thus provides a suitable testbed for crosstalk-suppression sequences.
The heavy-hex lattice alternates edge qubits (bivalent in the bulk) and vertex qubits (trivalent in the bulk).
We average over the edge qubits, colored red, green, and blue (RGB) in \cref{fig:heavy-hex-embedded-3coloring}(b,c) and referred to as the RGB qubits from here on. The vertex qubits are the spectators of the embedded $C=3$ coloring, and restricting the analysis to the RGB qubits allows a direct comparison with higher-$C$ sequences (see \cref{fig:chadd-equivalence-classes}).

We evaluated the ability of each sequence to preserve a set of initial states as follows:
prepare each RGB qubit in the same state, one of the canonical six Pauli eigenstates,
and separately prepare every other qubit in the same state, one of the three spectator preparation states $\ket{0}$, $\ket{+}$, or $\ket{1}$;
apply $m$ repetitions of the DD sequence, for evenly spaced integer $m$ up to the largest number of repetitions that fits within $T_\text{max}=38.4\,\mu\mathrm{s}=320\tau$, with $\tau=2\delta=120\,\mathrm{ns}$;
and, finally, undo the preparation rotations, which should ideally return each qubit to the ground state.
Each of these circuits was run $5000$ times.
Of the $6\times3=18$ preparation settings submitted, one, $\ket{{+}y}_{\rm RGB}\otimes\ket{1}_{\rm spec}$, failed mid-run; the same $17$ completed settings enter every sequence and every time point.
For each preparation setting, we averaged the ground-state outcomes over shots and over the qubits of the color subset of interest (the RGB qubits throughout, and single colors in \cref{fig:color-subsets}).
We plot the mean $\mu$ of the $17$ setting averages, the average fidelity of the repeated sequence, with error bars of $1\sigma$ in each direction, against the total time $mN\tau$, where $\sigma$ is the standard error of the mean over the $17$ preparation settings, i.e., their sample standard deviation divided by $\sqrt{17}$.
On the right vertical axis, the average pulse count of the RGB qubits is plotted against sequence time, and the slope of these lines is the qubit-weighted RGB PRR divided by $\tau$.

\subsection{State Preservation at \texorpdfstring{$C=3$}{C=3}}
\label{sec:c3-comparison}

Recall that $C=2$ CBDD and CGDD coincide with linear-depth CHaDD sequences (\cref{thm:cbdd,thm:cgdd}),
so $C=3$ is the smallest case in which the three DD families can be distinguished.
The $C=3$ sequences tested on \texttt{ibm\_strasbourg} are depicted as chromatic timeline diagrams featuring finite-width square pulses in \cref{fig:bivalent-3coloring-timeline-diagrams},
and the average fidelity of each sequence, taken over the RGB qubits and preparation settings, is depicted in \cref{fig:robustness-eliminates-prr-advantage}.
The tested CHaDD sequence uses the row-optimized assignment $g^*$ of \cref{sec:cwdd}, i.e., it is the $C=3$ CWDD sequence; as noted there, for $C=3$ the total pulse count is the same for every row assignment.
The idle evolution and the non-robust sequences of the left column of \cref{fig:bivalent-3coloring-timeline-diagrams} are plotted in \cref{fig:robustness-eliminates-prr-advantage}(a), and their robust versions, from the right column of \cref{fig:bivalent-3coloring-timeline-diagrams}, in \cref{fig:robustness-eliminates-prr-advantage}(b).

We see that the relative ordering of performance of the non-robust sequences in \cref{fig:robustness-eliminates-prr-advantage}(a) at short times is the reverse of the ordering of their PRR$/\tau$ slopes on the right vertical axis of \cref{fig:robustness-eliminates-prr-advantage}(b);
however, this performance ordering is flattened by the introduction of the robust sequences, which are designed to protect against coherent pulse imperfections.
This difference in performance in \cref{fig:robustness-eliminates-prr-advantage}(a), paired with the similarity in performance in \cref{fig:robustness-eliminates-prr-advantage}(b), supports treating the PRR as a proxy for coherent error accumulated through the DD sequence and as the key sequence selection metric: at equal performance, the sequence with the lowest PRR is favored.

Note that beyond $t\approx15\,\mu\mathrm{s}$, the minimum of every non-robust sequence in \cref{fig:robustness-eliminates-prr-advantage}(a) falls well below the idle curve, and the XX, CHaDD, and CBDD fidelities partially recover, as expected for accumulating coherent pulse error; the robust sequences in \cref{fig:robustness-eliminates-prr-advantage}(b) eliminate this behavior.
The minima occur at roughly equal cumulative pulse numbers
rather than at equal times [compare \cref{fig:color-subsets}(c)], consistent with a fixed coherent error per pulse.
This trend is explored in more detail across single-color qubit subsets in \cref{fig:color-subsets}, where it is even more pronounced, and across $C=3$ and $C=5$ sequences grouped by their RGB PRRs below [\cref{fig:chadd-equivalence-classes}].
Although not shown, the elimination of the ordering of sequence performance when moving to the robust versions was observed for each single-color qubit subset and can be seen in the data available online~\cite{brown2026cgdddatarepo}.

The similarity of the robust-sequence performance is consistent with all three schedules producing the same first-order average Hamiltonian.
The near-coincidence of UR4 with the robust chromatic sequences in \cref{fig:robustness-eliminates-prr-advantage}(b) further indicates that the next-nearest-neighbor couplings among the RGB qubits, which synchronous UR4 does not cancel, are too weak to produce a resolvable fidelity loss on this device.
One might instead attribute this agreement to the sequences sharing a common control group, but they do not: by the homomorphism analysis of \cref{sec:results}, minimum-depth $C=3$ CHaDD has a four-element control group, isomorphic to $\mathbb{F}_2^2$, whereas CBDD and CGDD share an eight-element control group, isomorphic to $\mathbb{F}_2^3$.
The first-order average requires only the row-sum and orthogonality conditions, not any relation among the control groups.

For the non-robust sequences, in contrast, performance is empirically determined essentially solely by the PRR of the relevant qubit subset: e.g., in \cref{fig:color-subsets}(a), CBDD and CGDD pulse the red qubits at the same rate (once every four time steps), and their red-subset fidelities overlap.
CHaDD pulses the red qubits every two time steps and XX every time step, and the red-subset fidelity worsens each time the red PRR doubles.
The same trend is observed for the green qubits in \cref{fig:color-subsets}(b), where the only difference in the schedules is that CBDD pulses the green qubits every two time steps rather than every four, so that CBDD overlaps with CHaDD instead of CGDD.
In \cref{fig:color-subsets}(c), XX, CHaDD, and CBDD all pulse the blue qubits every time step, and their performance overlaps.
Across the panels of \cref{fig:color-subsets}, disjoint qubit subsets pulsed at the same rate exhibit similar fidelities.

The absence of oscillations in the robust sequences disfavors pulse interference as the dominant source of the performance differences between the non-robust sequences.
If closely spaced pulses produced substantial interference, the robust versions would also be expected to exhibit oscillations~\cite{vezvaee2024virtualzgatessymmetric}.
The observations we report here therefore support coherent pulse error as the dominant source, although they do not by themselves exclude every pulse-interference contribution.

\begin{figure*}[t]
\centering
\begin{minipage}[t]{0.386\textwidth}
(a)\\
\includegraphics[width=\linewidth]{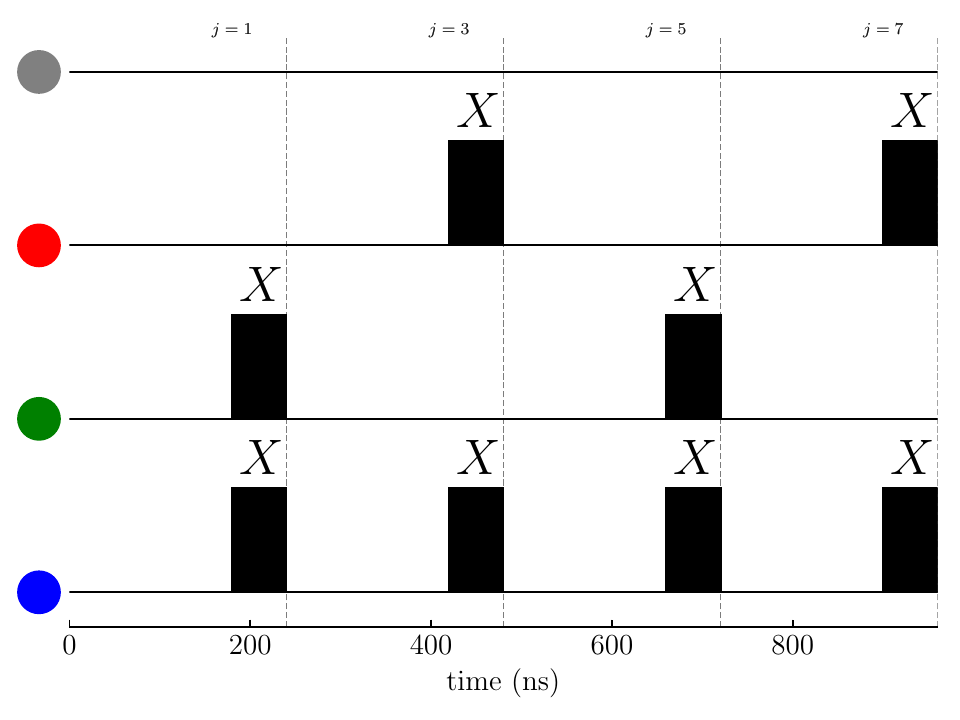}\\
(b)\\
\includegraphics[width=\linewidth]{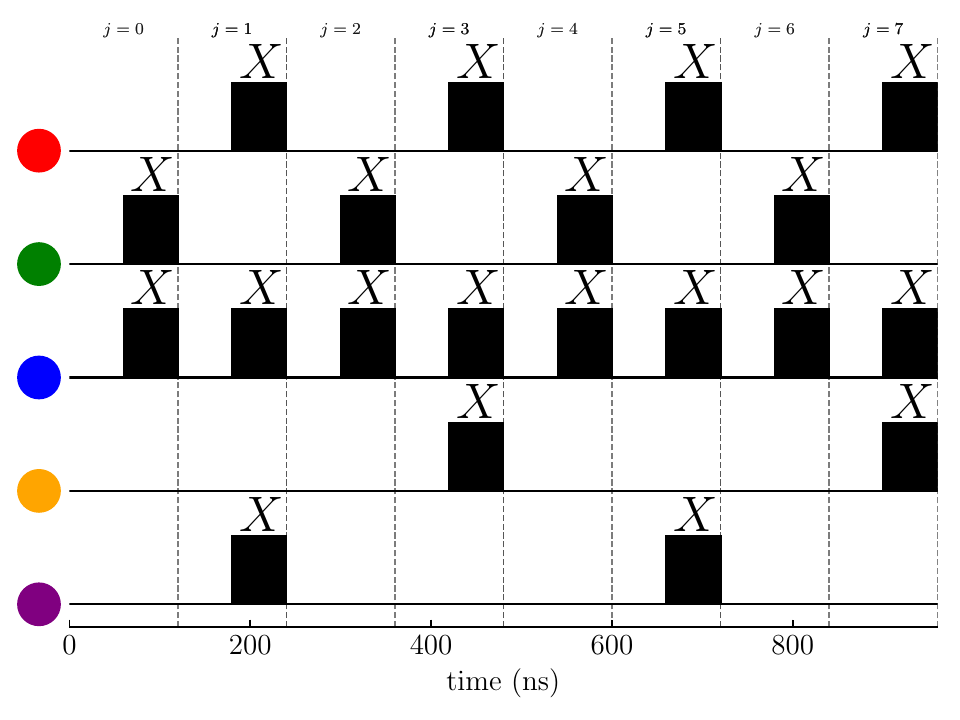}\\
(c)\\
\includegraphics[width=\linewidth]{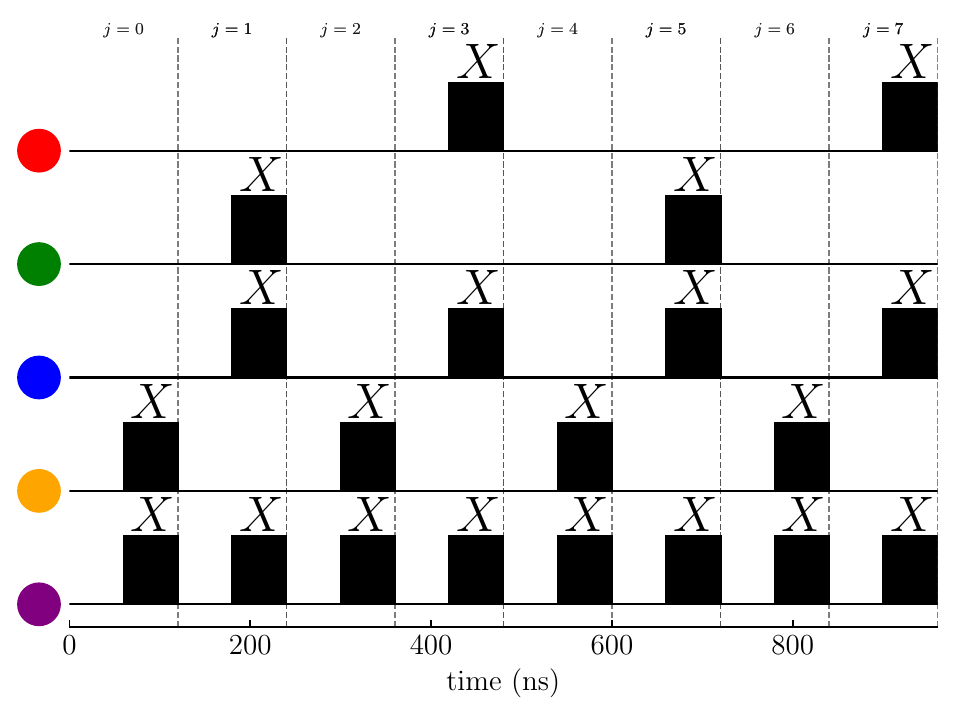}
\end{minipage}\hfill
\begin{minipage}[t]{0.594\textwidth}
(d)\\
\includegraphics[width=\linewidth]{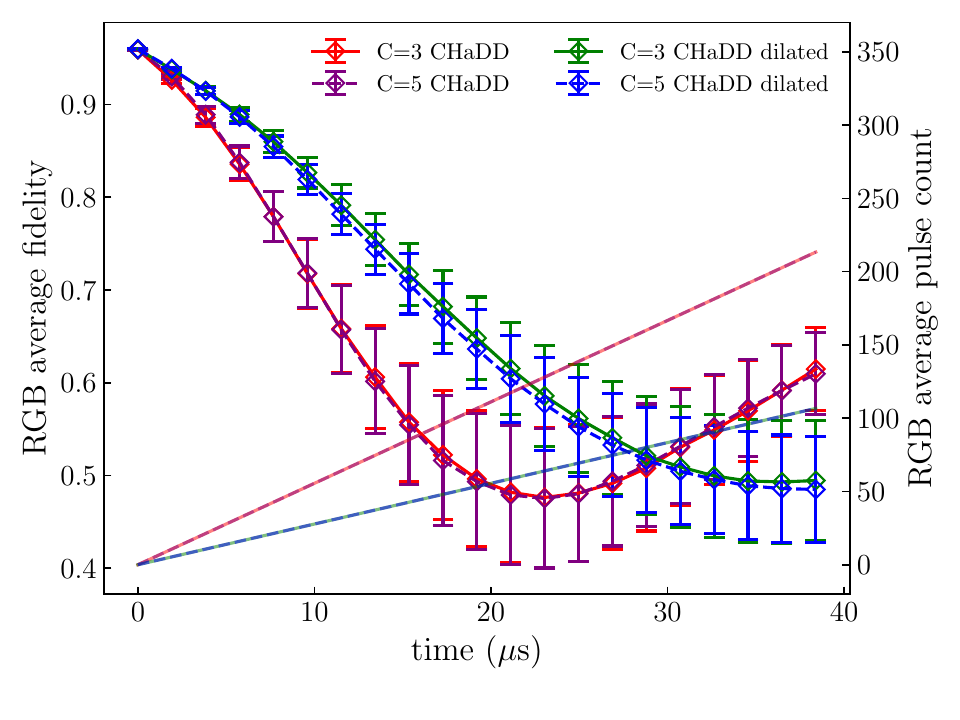}\\
(e)\\
\includegraphics[width=\linewidth]{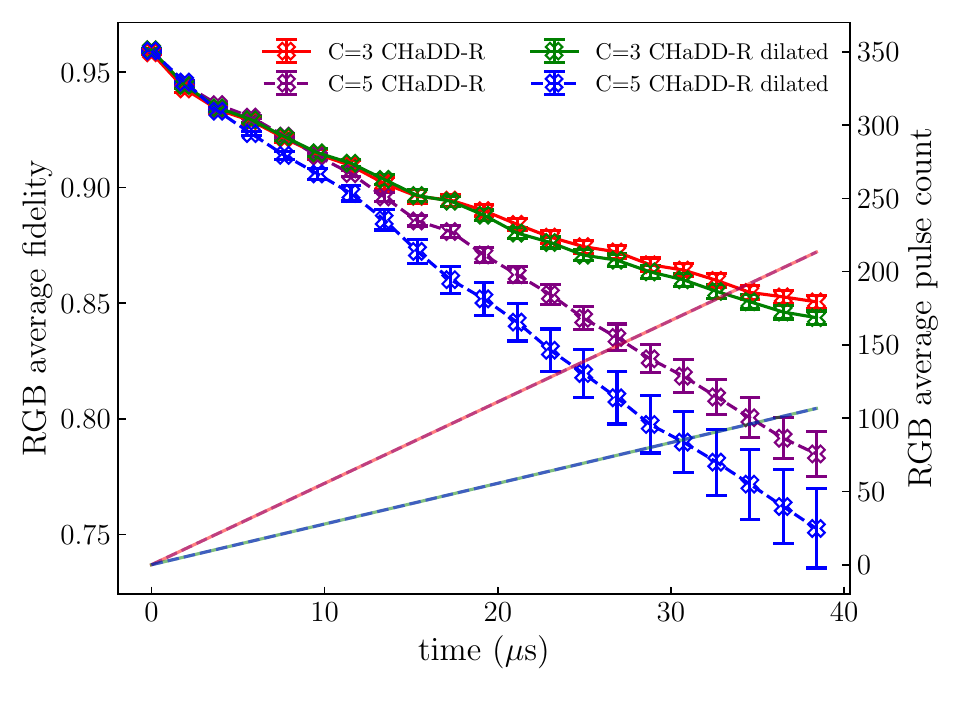}
\end{minipage}
\caption{
CHaDD equivalence classes grouped by RGB PRR;
the basic $C=3$ CHaDD sequence, not repeated here, appears in the second row of the left column of \cref{fig:bivalent-3coloring-timeline-diagrams} [coloring in \cref{fig:heavy-hex-embedded-3coloring}(c)].
(a) $C=3$ CHaDD dilated;
(b) $C=5$ CHaDD [coloring in \cref{fig:heavy-hex-embedded-3coloring}(b)] sharing the RGB rows of the basic $C=3$ sequence;
(c) $C=5$ CHaDD with dilated RGB rows;
(d) RGB average fidelity and average pulse count plotted as a function of time for the basic $C=3$ sequence and the three sequences in panels (a-c);
(e) Same with robust versions of sequences.
For (d) and (e), the legend lists the basic $C=3$ CHaDD sequence followed by the sequences of panels (a-c), and $C=3$ has solid lines while $C=5$ has dashed lines.
}
\label{fig:chadd-equivalence-classes}
\end{figure*}

\begin{figure}[p]
\centering
\setlength{\subfigtopskip}{2pt}\setlength{\subfigcapskip}{1pt}\setlength{\subfigbottomskip}{2pt}
\subfigure[\ ]{\includegraphics[width=0.94\columnwidth]{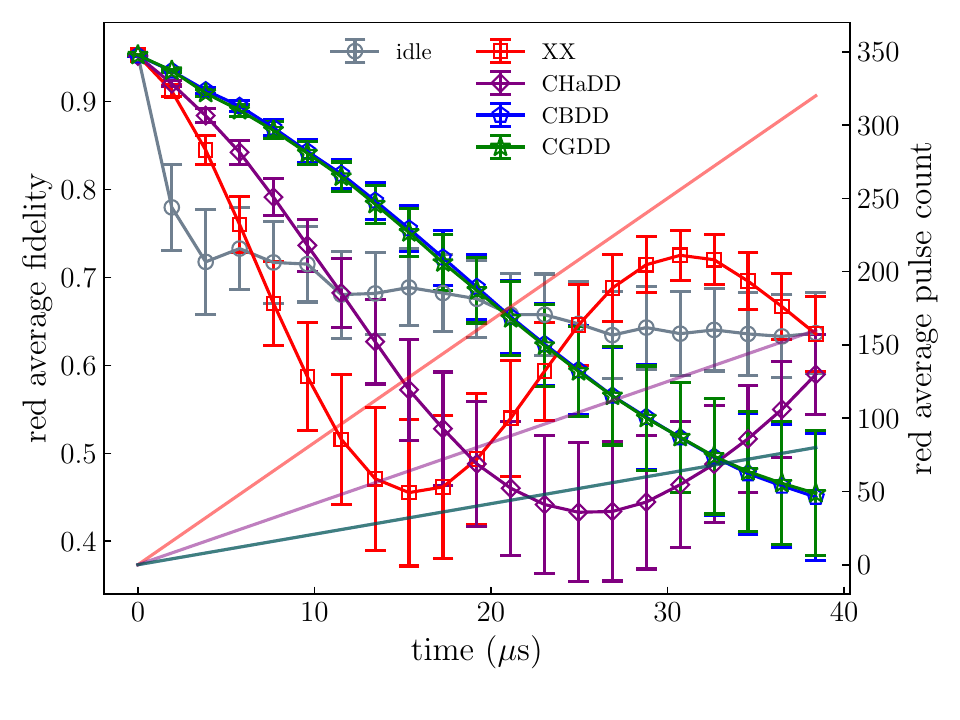}}
\subfigure[\ ]{\includegraphics[width=0.94\columnwidth]{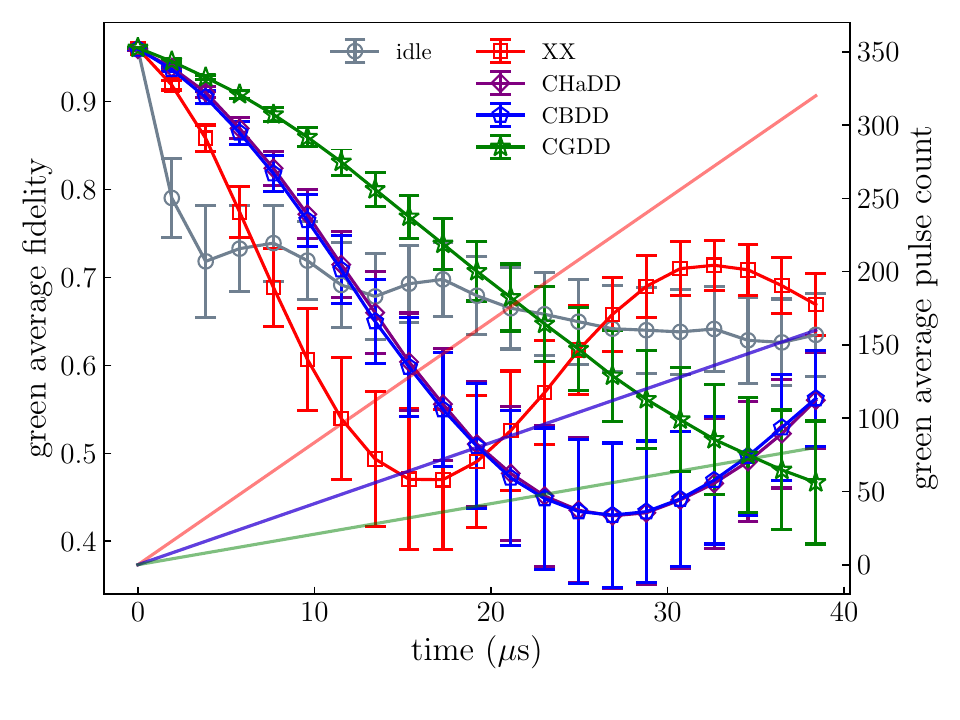}}
\subfigure[\ ]{\includegraphics[width=0.94\columnwidth]{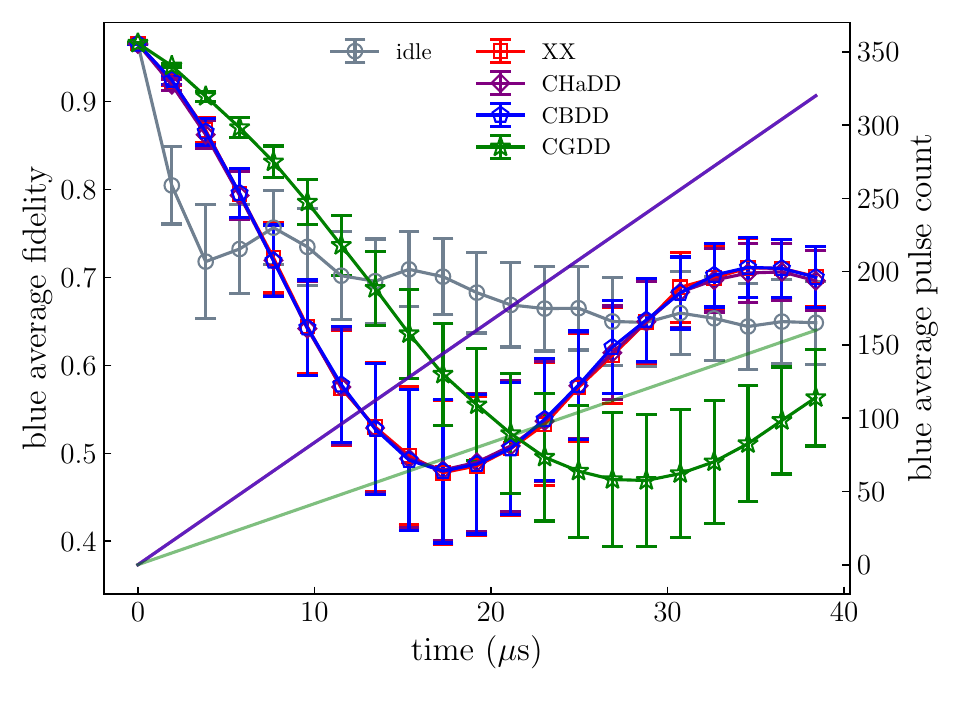}}
\caption{
Average fidelity and pulse count across different color subsets:
(a) red, (b) green, and (c) blue qubits.
Sequences that pulse the qubits of the displayed color at the same rate overlap: on the red qubits (a), CBDD and CGDD (a pulse every four time steps); on the green qubits (b), CHaDD and CBDD (every two time steps); and on the blue qubits (c), XX, CHaDD, and CBDD (every time step).
Curves with equal pulse spacing are moreover similar across panels.
}
\label{fig:color-subsets}
\end{figure}

\subsection{CHaDD Equivalence Classes from Distance-1 3-coloring and Distance-3 5-coloring}
\label{sec:equivalence-classes}

To explore how equivalence classes of DD sequence performance emerge across different numbers of colors $C$, we choose two values for which colorings are shown in \cref{fig:heavy-hex-embedded-3coloring}(c) and (b): $C=3$ and $C=5$, respectively, which share red, green, and blue qubits in common.
Starting with the basic $C=3$ CHaDD sequence, whose timeline diagram appears in the second row of the left column of \cref{fig:bivalent-3coloring-timeline-diagrams}, we create a $C=5$ CHaDD sequence that copies the red, green, and blue schedule (and therefore has the same RGB PRR) and adds orthogonal DD sequences (i.e., DD sequences arising from orthogonal rows of sign matrices) to the orange and purple qubits, which are gray in the $C=3$ coloring; the resulting timeline diagram is depicted in \cref{fig:chadd-equivalence-classes}(b).

We form the dilated $C=3$ schedule of \cref{fig:chadd-equivalence-classes}(a) by inserting one idle time step before every original time step of the basic $C=3$ CHaDD sequence, at fixed $\tau$; this doubles $N$ while leaving the pulse count $P$ unchanged and therefore halves the PRR on every qubit or color subset.
Panel (c) shows a second $C=5$ CHaDD row assignment with the same depth $N=8$ as panel (b): its RGB rows are the dilated rows of panel (a) (rows 4, 6, and 2 of $W_3$), while the orange and purple qubits are pulsed every two time steps and every time step, respectively (rows 3 and 1), so its RGB PRR equals that of panel (a).
Note that panels (a) and (c) share the same cycle time $N\tau$, so their comparison isolates the PRR of the neighboring qubits [\cref{eq:subset-prr}].
The state-preservation performance of these sequences and their robust counterparts is plotted in \cref{fig:chadd-equivalence-classes}(d) and (e), respectively, together with the RGB pulse count on the right vertical axis.
In \cref{fig:chadd-equivalence-classes}(d), the similar performance of the $C=3$ and $C=5$ CHaDD schedules is consistent with their equal RGB PRR by design; the same correspondence holds for the pair of panels (a) and (c), whose common RGB PRR is half as large.
Thus, even for different values of $C$, equal RGB PRR is empirically associated with similar state-preservation performance for non-robust sequences.

The robust sequences in \cref{fig:chadd-equivalence-classes}(e) eliminate the disparity between the two $C=3$ sequences but not between the $C=5$ sequences, which perform worse, and worst when the orange and purple qubits are pulsed most often [panel (c)].
Thus, pulses applied to the neighbors of the RGB qubits degrade RGB state preservation in a way that the robust single-qubit construction does not compensate.
This is consistent with crosstalk generated during the finite-width pulses, which lies outside the instantaneous-pulse analysis, and with the fact that replacing every second and third $X$ pulse with an $\overline{X}$ pulse is merely an ansatz for constructing robust versions of chromatic DD sequences rather than a general robustness proof.

In summary, distance-$d$ colorings with $d>1$ extend protection to couplings of range up to $d$ at the cost of more colors, while for the non-robust sequences tested here, equal PRR on the qubits of interest is empirically associated with similar state preservation across $d$ and $C$.
The PRR serves as a proxy for accumulated coherent error, but robust sequences largely remove this association, up to the residual effect of pulses applied to neighboring qubits.

\section{Conclusions and outlook}
\label{sec:conc}
To summarize, we have introduced CBDD and CGDD, two special cases of CHaDD with non-minimum depth for $C>2$, which preserve initial states better than CHaDD at short times in our non-robust-pulse experiments.
The data support accumulated coherent pulse error as the dominant explanation for the observed performance differences.
We have proved the first-order decoupling properties, circuit depths, and pulse repetition rates of CHaDD, CBDD, and CGDD (\cref{thm:chadd,thm:cbdd,thm:cgdd}) and clarified the relationships among their control groups.
The similarity of their robust counterparts is consistent with their common first-order average Hamiltonian once coherent pulse errors are suppressed.

Although CBDD and CGDD have exponentially rather than linearly scaling circuit depth in $C$, their lower PRRs can make them preferable in non-robust implementations when the longer sequence fits within the available idle interval.
In the experiments reported here, the robust variants perform better overall than their non-robust counterparts.
Moreover, in the robust setting, where CHaDD, CBDD, and CGDD exhibit on par performance, the sequence with the lowest PRR is preferred as it is the least resource intensive; by this standard, the robust CGDD sequence emerges as the top performer.

Still missing is a formal robustness proof for the robust chromatic sequences: replacing the second and third $X$ pulse in each block of four by $\overline{X}$ is an ansatz borrowed from the single-qubit UR$_n$, $n=4$ sequence~\cite{Genov2017}, whose error scaling for arbitrary even sequence length $n$ has recently been proven~\cite{dalessandro2026prooferrorscalinguniversally}. The differences among the robust sequences in \cref{fig:chadd-equivalence-classes} are correlated with the PRR on the qubits neighboring the RGB set, consistent with the ansatz not compensating the effect of pulses applied to those qubits.

Higher-order chromatic sequences exist in the weak-coupling regime: Ref.~\cite{Kim2026highorder} proves that for any group $\mathcal{G}$ that averages the system-bath interaction to zero, at most $(|\mathcal{G}|-1)p$ multi-qubit pulse events at nonuniform intervals suffice to cancel every error term linear in the coupling strength through order $p$ in the total evolution time, and applies this result to CHaDD.
This bound counts pulse events, each of which can pulse several colors at once, whereas the PRR counts the pulses each qubit undergoes at a fixed time step $\tau$, so the pulse-event bound does not determine the PRR.
Combining higher-order cancellation with low PRR and robust pulses, multi-axis chromatic sequences, higher-order chromatic sequences obtained by other constructions (e.g., by concatenation~\cite{Khodjasteh:2005xu} or nested Uhrig sequences~\cite{Wang:10,Xia:2011uq}), experiments at larger $C$ and on other QPUs, and the relation between the sign-matrix conditions used here and the code-based constructions of Ref.~\cite{Nguyen2026color} are natural next steps.
We expect the PRR to remain the relevant figure of merit whenever coherent pulse errors, rather than higher-order decoupling errors, dominate.

\acknowledgments
This material is based upon work supported by, or in part by, the U.S. Army Research Laboratory and the U.S. Army Research Office under contract/grant number W911NF2310255.
This research was supported by the Office of the Director of National Intelligence (ODNI), Intelligence Advanced Research Projects Activity (IARPA) and the Army Research Office, under the Entangled Logical Qubits program through Cooperative Agreement Number W911NF-23-2-0216. We acknowledge funding from the Office of Naval Research under Grant N00014-26-1-2092.

\section*{Data Availability}
The data presented in this paper and the notebooks needed to recreate the plots are available online~\cite{brown2026cgdddatarepo}.

\appendix

\section{Proof of the CHaDD Theorem}
\label{app:CHaDD-thm-proof}

The proof is essentially identical to the one given in Ref.~\cite{brown2024efficient}. 

\begin{proof}[Proof of \cref{thm:chadd}]
Let $G=(V,E)$ be a graph representing a 2-local Hamiltonian $H_G$ to be decoupled, and let $f$ be a proper $C$-coloring of $G$; here $E$ includes every coupling in \cref{eq:2-body-hamiltonian}, intentional or not, so the coloring is proper with respect to all of them.
The bath operators elided in the color-partitioned form of \cref{sec:noise-model} commute with the system-only control unitaries and therefore carry through the first-order averaging below unchanged.

For equal free-evolution intervals and instantaneous $X$ pulses, the $j$th toggling-frame interval propagator is
\begin{align}
    \mathcal{F}_j
    \equiv U_j^\dagger f_\tau U_j
    = \exp\left(-i\tau U_j^\dagger H_G U_j\right),
\end{align}
where $U_j$ is defined in \cref{eq:chadd-decoupling-unitary}.

A single-body term $\sigma_v^\alpha$ anticommutes with $\widetilde{X}_c$ if $f(v)=c$ and $\alpha \neq x$, so it is conjugated by $U_j$ as
\bes
\begin{align}
    \widetilde{X}_c^\dagger \sigma_v^\alpha \widetilde{X}_c
    &=
    (-1)^{ \delta_{c,f(v)} \left(1-\delta_{\alpha x}\right) } \sigma_v^\alpha
    , \\
    U_j^\dagger \sigma_v^\alpha U_j
    &=
    (-1)^{ \left\{g\left[f(v)\right] \cdot j \right\} \left(1-\delta_{\alpha x}\right) } \sigma_v^\alpha
    .
\end{align}
\ees
Thus, the 1-body Hamiltonian $H_1$ [\cref{eq:1-body-hamiltonian}] is conjugated by $U_j$ as
\bes
\begin{align}
    U_j^\dagger H_1 U_j
    &=
    \sum_\alpha \sum_{c=1}^C (-1)^{ \left[g(c) \cdot j \right] \left(1-\delta_{\alpha x}\right) } \sum_{v \in V_c} \sigma_v^\alpha
    \\&=
    H_1^x + \sum_{\alpha \neq x} \sum_{c=1}^C (-1)^{g(c) \cdot j} \sum_{v \in V_c} \sigma_v^\alpha
    .
\end{align}
\ees
Averaging over time steps $j=0,\dots,N-1$, we obtain the decoupling-averaged 1-body Hamiltonian
\bes
\begin{align}
    \overline{H}_1
    &=
    \frac1N \sum_{j=0}^{N-1} U_j^\dagger H_1 U_j
    \\&=
    H_1^x +
    \sum_{\alpha \neq x} \sum_{c=1}^C \frac1N \sum_{j=0}^{N-1} (-1)^{g(c) \cdot j}  \sum_{v \in V_c} \sigma_v^\alpha
    \\&=
    H_1^x
    .
\end{align}
\ees
The only 1-body term that survives the first-order average is $H_1^x$.
Every other $\sigma_v^\alpha$ term with $f(v)=c$ and $\alpha\neq x$ is modulated by $(-1)^{g(c) \cdot j}$, whose sum over $j$ vanishes because $g(c)\neq0$.

A two-body term $\sigma_u^\alpha \sigma_v^\beta$ with $f(u)=c_1 \neq c_2=f(v)$ is conjugated by $U_j$ as
\begin{align}
    U_j^\dagger \sigma_u^\alpha \sigma_v^\beta U_j
    &=
    (-1)^{ \left[g(c_1) \cdot j \right] \left(1-\delta_{\alpha x}\right) }
    (-1)^{ \left[g(c_2) \cdot j \right] \left(1-\delta_{\beta x}\right) }
    \sigma_u^\alpha \sigma_v^\beta
    .
\end{align}
Thus, $\sigma_u^x \sigma_v^x$ commutes with every $U_j$;
for $\alpha\neq x$, $\sigma_u^\alpha \sigma_v^x$ is modulated by $(-1)^{g(c_1) \cdot j}$;
for $\beta\neq x$, $\sigma_u^x \sigma_v^\beta$ is modulated by $(-1)^{g(c_2) \cdot j}$;
and, for $\alpha,\beta\neq x$, $\sigma_u^\alpha \sigma_v^\beta$ is modulated by
$(-1)^{g(c_1) \cdot j} (-1)^{g(c_2) \cdot j}$.
Since $g$ is injective, $c_1 \neq c_2$ implies $g(c_1) \neq g(c_2)$.
The corresponding Hadamard rows are therefore distinct and orthogonal:
\begin{align}
    \frac1N \sum_{j=0}^{N-1}
    (-1)^{g(c_1)\cdot j}(-1)^{g(c_2)\cdot j}
    =0
    .
\end{align}
Consequently, all but the $\sigma_u^x \sigma_v^x$ terms are averaged away, so the decoupling-averaged 2-body Hamiltonian is $\overline{H}_2=H_2^{xx}$ and
\begin{align}
    \overline{H}_G
    = \overline{H}_1 + \overline{H}_2
    = H_1^x + H_2^{xx}
    .
\end{align}

Let $T=N\tau$.
The Baker-Campbell-Hausdorff expansion now connects this average Hamiltonian to the physical cycle propagator:
\begin{align}
    \mathcal{F}_{N-1}\cdots\mathcal{F}_0
    &=
    \exp\left(-iT\overline{H}_G\right)+O\left(T^2\|H_G\|^2\right)
    \notag\\
    &=
    \exp\left[-iT\left(H_1^x+H_2^{xx}\right)\right]+O\left(T^2\|H_G\|^2\right)
    .
\end{align}
This establishes the claimed cancellation to first order in $\tau$.
Note that this is a fixed-$N$ estimate: the $O(T^2\|H_G\|^2)$ remainder bound is small only when the cycle time $T=N\tau$ satisfies $T\|H_G\|\ll1$, a condition that becomes more stringent as $N$ grows exponentially with $C$ for CBDD and CGDD.
Next,
\begin{align}
    C < N
    =2^\nu
    =2^{\lfloor\log_2 C\rfloor+1}
    \leq 2C
    ,
\end{align}
which proves the circuit-depth bound.

It remains to establish the PRR bounds.
By \cref{eq:prr-optimal-chadd-pulse-count}, the row assignment $g^*$ of \cref{sec:cwdd} minimizes the total pulse count per repetition, $P=\frac{C(C+1)}2+\lceil\frac C2\rceil$; combined with $C<N\leq2C$, this yields \cref{eq:prr-optimal-chadd}.
\end{proof}

\bibliographystyle{apsrev4-2}

\end{document}